\documentclass[11pt, a4paper]{article}

\usepackage[utf8]{inputenc}
\usepackage[T1]{fontenc}
\usepackage{lmodern}
\usepackage{microtype}
\usepackage[pdfborder={0 0 0}]{hyperref}

\usepackage{color}
\usepackage{marginnote} 
\usepackage{ifthen} 
\usepackage{paralist} 
\usepackage{graphicx}

\usepackage[tbtags]{amsmath} 
\usepackage{amssymb}
\usepackage{amsthm}
\usepackage{dsfont} 
\usepackage{mathtools} 
\usepackage{tensor}

\usepackage[text={15.5cm, 23.0cm}, centering]{geometry}
\newcommand{\cif}{C^\infty}
\newcommand{\ee}[0]{\varepsilon}
\newcommand{\inv}[0]{{-1}}

\newcommand{\boldify}[1]{\ensuremath{\boldsymbol{#1}}}

\def\ba{\boldify{a}}

\def\be{\boldify{e}}

\def\bj{\boldify{j}}

\def\bm{\boldify{m}}

\def\bq{\boldify{q}}

\def\bt{\boldify{t}}

\def\bw{\boldify{w}}
\def\bx{\boldify{x}}
\def\by{\boldify{y}}

\newcommand{\CC}{\mathbb{C}}
\newcommand{\NN}{\mathbb{N}}
\newcommand{\RR}{\mathbb{R}}

\newcommand{\algebra}[1]{\mathfrak{#1}}
\newcommand{\ag}{\algebra{g}}
\newcommand{\ah}{\algebra{h}}
\newcommand{\csa}{\algebra{h}}
\newcommand{\iso}{\algebra{iso}}
\newcommand{\so}{\algebra{so}}

\newcommand{\idad}[1]{\bigl(\mathds{1}-\Ad(u_{#1})\bigr)}

\newcommand{\tidad}[2]{\tensor{\idad{#1}}{#2}}
\newcommand{\idadi}[1]{\bigl(\mathds{1}-\Ad(u_{#1}^{-1})\bigr)}

\newcommand{\tidadi}[2]{\tensor{\idadi{#1}}{#2}}

\newcommand{\defeq}{\coloneqq}

\newcommand{\isoeq}{\cong}

\newcommand{\weaklyequal}{\approx}

\newcommand{\diffd}{\mathrm{d}}

\newcommand{\tdiff}[2]{\frac{\diffd #1}{\diffd #2}}
\newcommand{\tdiffat}[3]{\left.\tdiff{#1}{#2}\right|_{#3}}

\newcommand{\ThetaFam}[3]{\Theta^{#2,#1}_#3}

\newcommand{\oo}{\otimes}

\DeclareMathOperator{\Ad}{Ad}
\DeclareMathOperator{\diag}{diag}
\DeclareMathOperator{\Hom}{Hom}
\DeclareMathOperator{\ISO}{ISO}
\DeclareMathOperator{\Mat}{Mat}
\DeclareMathOperator{\PSL}{PSL}
\DeclareMathOperator{\SL}{SL}
\DeclareMathOperator{\SO}{SO}
\DeclareMathOperator{\Span}{span}
\DeclareMathOperator{\Tr}{Tr}

\newcommand{\loal}{\ensuremath{\so(2,1)}}
\newcommand{\poal}{\ensuremath{\iso(2,1)}}
\newcommand{\logr}{\ensuremath{\SO_+(2,1)}}
\newcommand{\pogr}{\ensuremath{\ISO(2,1)}}
\newcommand{\pogrforsection}{ISO(2, 1)}
\newcommand{\allpogr}{\ensuremath{\pogr^{n+2g}}}
\newcommand{\alllogr}{\ensuremath{\logr^{n+2g}}}
\newcommand{\restpogr}{\ensuremath{\pogr^{n-2+2g}}}
\newcommand{\restlogr}{\ensuremath{\logr^{n-2+2g}}}

\newcommand{\frbivector}{\ensuremath{B_r^{n,g}}}
\newcommand{\frrestbivector}[1][]{\ensuremath{B_{\ifthenelse{\equal{#1}{}}{r}{#1}}^{n-2,g}}}

\newcommand{\equivalencequotientspace}[1]{#1/\!\!\sim}

\newtheorem{theorem}{Theorem}[section]
\newtheorem*{theorem*}{Theorem}

\newtheorem{lemma}[theorem]{Lemma}

\newtheorem{corollary}[theorem]{Corollary}
\newtheorem{definition}[theorem]{Definition}

\newcommand{\phasespace}{\ensuremath{\mathcal{P}}}
\newcommand{\extphasespace}{\ensuremath{\mathcal{P}_{\mathrm{ext}}}}
\newcommand{\symplextphasespace}{\ensuremath{\mathcal{Q}}}
\newcommand{\csurface}{\ensuremath{\Sigma}}
\newcommand{\oncsurface}{\vert_\csurface}

\newcommand{\definee}[1]{\textbf{#1}}

\newcommand{\allowlinebreakhere}{\linebreak[0]}

\newcommand{\email}[1]{\href{mailto:#1}{\nolinkurl{#1}}}

\newcommand{\ie}{i.e.\ }

\def\mytitle{Gauge fixing and classical dynamical r-matrices in \pogrforsection-Chern-Simons theory}
\def\myauthors{C.~Meusburger and T.~Sch\"onfeld}
\hypersetup{pdftitle=\mytitle, pdfauthor=\myauthors}

\begin{document}

\begin{center}
  {\huge\mytitle}

  \vspace{2em}

  {\large
   C.~Meusburger\footnote{\email{catherine.meusburger@math.uni-erlangen.de}} \qquad\qquad
   T.~Sch\"onfeld\footnote{\email{torsten.schoenfeld@math.uni-erlangen.de}}}

  \vspace{1em}

  Department Mathematik, \\
  Friedrich-Alexander  Universit\"at Erlangen-N\"urnberg \\
  Cauerstra\ss e 11, 91058 Erlangen, Germany

  \vspace{1em}

  March 23, 2012

  \vspace{2em}

  \begin{abstract}
    We apply the Dirac gauge fixing procedure to Chern-Simons theory with gauge
    group $\pogr$ on manifolds $\mathbb R\times S$, where $S$ is a punctured
    oriented surface of general genus. For all gauge fixing conditions that
    satisfy certain structural requirements, this yields an explicit
    description of the Poisson structure on the moduli space of flat
    $\pogr$-connections on $S$
    in terms of classical dynamical $r$-matrices for $\poal$.  We show that the
    Poisson structures and classical dynamical $r$-matrices arising from
    different gauge fixing conditions are related by dynamical $\pogr$-valued
    transformations that generalise the usual gauge transformations of
    classical dynamical $r$-matrices.  By means of these transformations, it is
    possible to classify all Poisson structures and classical dynamical
    $r$-matrices obtained from such gauge fixings. Generically these Poisson
    structures combine classical dynamical $r$-matrices for non-conjugate
    Cartan subalgebras of $\poal$.
  \end{abstract}
\end{center}

\section{Introduction}

Moduli spaces of flat connections on punctured Riemann surfaces and their
quantisation are of interest to both mathematics and physics due to their rich
mathematical structure and their links with a variety of topics from geometry,
algebra and gauge theory.  From the physics perspective, a major motivation for
their study is their role in Chern-Simons theory. Moduli spaces of flat
connections can be viewed as the gauge-invariant or reduced phase spaces of
Chern-Simons theories. Their quantisation is thus related to structures arising
in the quantisation of Chern-Simons theory such as quantum groups, aspects of
knot theory \cite{Witten:1989aa} and topological quantum field
theories. Another important feature of the theory is its relation to
three-dimensional gravity \cite{Achucarro:1986aa, Witten:1988aa,Witten:1989sx}.

The quantisation of moduli spaces of flat $G$-connections and their relation to
quantum Chern-Simons theory are well understood for the case of compact,
semisimple Lie groups $G$.  In this setting, quantisation can be achieved via
many formalisms, and most of these formalisms involve the representation theory
of a quantum group, namely the $q$-deformed universal enveloping algebra
$U_q(\ag)$ at a root of unity.  In the case of non-compact non-semisimple Lie
groups, the quantisation proves more difficult.  Although there are partial
results on the quantisation of these cases via analytic continuation
\cite{Witten:2011aa} and results for specific Lie groups \cite{Witten:1989ip,
  Buffenoir:2002aa, Meusburger:2004aa, Meusburger:2010aa, Meusburger:2010bb},
there is currently no general framework to address this case.

From the viewpoint of Hamiltonian quantisation formalisms, these difficulties
are related to the fact that Chern-Simons theory and the associated moduli
spaces of flat connections can be viewed as constrained Hamiltonian systems. In
the non-compact setting, representation-theoretical complications lead to
difficulties in the implementation of the constraint operators in the quantum
theory.  It therefore seems advisable to also consider other approaches to the
quantisation of moduli spaces of flat connections with non-compact gauge
groups.  This includes in particular ``quantisation after constraint
implementation'' approaches, which proceed by first applying the Dirac gauge
fixing formalism to the classical theory and then quantising the resulting
Poisson structure.  However, besides partial results
for $\SL(2,\CC)$-Chern-Simons theory \cite{Buffenoir:2003aa, Buffenoir:2005aa,
  Noui:2005aa}, this avenue has not been pursued yet.

An independent mathematical motivation for investing gauge fixing procedures
related to moduli spaces of flat connections arises from Poisson geometry. Such
gauge fixing procedures can be interpreted as the Poisson counterpart of
symplectic reduction. Moduli spaces of flat connections played an important
role in many interesting developments in this subject such as Lie-group-valued
moment maps, see \cite{Alekseev:1998aa} and the references therein. Moreover,
the symplectic structure on the moduli space can be characterised in terms of
certain Poisson structures from the theory of Poisson-Lie groups. Gauge fixing
in this context has been shown to give rise to classical dynamical $r$-matrices
in some cases \cite{Feher:2001aa, Feher:2004aa}.

In this article, we undertake a systematic investigation of Dirac gauge fixing
for the moduli space of flat $\pogr$-connections on a Riemann surface $S_{g,n}$
of general genus $g$ and with $n\geq 2$ punctures.  Our choice of the group
$\pogr=\logr\ltimes\mathbb R^3$ is motivated by the fact that it is an example
of a non-compact non-semisimple Lie group and that Chern-Simons theory with
gauge group $\pogr$ is closely related to (2+1)-gravity
\cite{Achucarro:1986aa,Witten:1988aa}.

Our starting point is a description of the moduli space of flat
$\pogr$-connections on $S_{g,n}$ in terms of a Poisson structure on the direct
product $\extphasespace=\allpogr$ due to Alekseev and Malkin
\cite{Alekseev:1995ab} and Fock and Rosly \cite{Fock:1998aa}.  It is shown in
\cite{Alekseev:1995ab, Fock:1998aa} that this Poisson structure is given in
terms of certain Poisson structures related to Poisson-Lie groups and after
reduction induces the canonical Poisson structure on the moduli space.  In the
case of the group $\pogr$, this Poisson structure is given by the direct
product of $n$ copies of the dual Poisson-Lie structure on $\pogr$ and $2g$
copies of the cotangent bundle Poisson structure for $\logr$
\cite{Meusburger:2004aa}:
\begin{equation*}
  \extphasespace=\underbrace{\pogr^*\times\ldots\pogr^*}_{n\;\times}\times \underbrace{T^*(\logr)\times\ldots\times T^*(\logr)}_{2g\;\times}.
\end{equation*}
From the physics perspective, the Poisson manifold $(\extphasespace,
\{\;,\;\})$ can be viewed as a constrained system with a set of six first-class
constraints from which the moduli space and its symplectic structure are
obtained after constraint implementation. This allows one to choose appropriate
gauge fixing conditions and to apply the Dirac gauge fixing procedure
\cite{Dirac:1949aa,Dirac:1950aa} to this Poisson structure. We explicitly
compute the resulting Poisson bracket for a large class of gauge fixing
conditions and investigate the resulting Poisson structures. This yields our
first central result:

\begin{theorem*}
  For suitable gauge fixing conditions, the Dirac gauge fixing procedure
  applied to $(\extphasespace,\{\;,\;\})$ gives rise to a Poisson structure on
  $\RR^2\times \restpogr$ which is determined uniquely by a solution of the
  classical dynamical Yang-Baxter
  equation.
\end{theorem*}

In particular, we find that this Poisson structure on $\RR^2\times\restpogr$ is
given by a formula directly analogous to the original Poisson structure on
$\extphasespace$. The only difference is that the classical $r$-matrix in the
original definition is replaced by a solution of the classical dynamical
Yang-Baxter equation for $\poal$ whose two dynamical variables parametrise
$\RR^2\subset \RR^2\times\restpogr$.

We then investigate the relation between the Poisson structures and solutions
of the classical dynamical Yang-Baxter equation that result from different
choices of gauge fixing conditions.  This leads to our second central result:
\begin{theorem*}
  All solutions of the classical dynamical Yang-Baxter equation obtained from
  gauge fixing are related by dynamical $\pogr$-valued transformations that
  generalise the gauge transformations of classical dynamical $r$-matrices from
  \cite{Etingof:1998aa}. All such solutions are locally equivalent to one of
  two standard solutions corresponding to Cartan subalgebras of $\poal$.
\end{theorem*}

The second statement of the theorem refers to an interesting feature of our
solutions that does not appear to have been observed in the literature yet. The
solutions of the classical dynamical Yang-Baxter equation that arise from
generic gauge fixing conditions are not associated with a fixed Cartan
subalgebra of $\poal$ but combine classical dynamical $r$-matrices for two
non-conjugate Cartan subalgebras of
$\poal$.

Our paper is structured as follows. Section \ref{sec:notation} contains the
basic definitions, notation and conventions used in the remainder of the
paper. Section \ref{sec:physics} summarises the physics background and
motivation of this work.  It can be skipped without loss by the reader
unfamiliar with this or interested mainly in the mathematical results.  It
contains a brief discussion of Chern-Simons theory on manifolds of topology
$\RR \times S_{g,n}$ and of the moduli space of flat connections on $S_{g,n}$
as a reduced or gauge-invariant phase space of Chern-Simons theory. It then
explains the Dirac gauge fixing formalism and its relation to symplectic
reduction and discusses the constraints and gauge fixing conditions imposed to
obtain the moduli space.

Section \ref{sec:gaugefixing} contains the first central result of this
article, namely the explicit description of the Poisson structure resulting
from Dirac gauge fixing for a general set of gauge fixing conditions.  We show
that the resulting Poisson structures are associated with solutions of the
classical dynamical Yang-Baxter equation and discuss examples arising from
specific choices of gauge fixing conditions as well as two simple standard
solutions associated with Cartan subalgebras of $\poal$.

In Section \ref{sec:dynamic-poincare-trafos} we introduce dynamical
$\pogr$-transformations which can be viewed as transitions between different
gauge fixing conditions. We determine the associated transformations of the
Dirac bracket and show how they can be interpreted as transitions between
different solutions of the classical dynamical Yang-Baxter equation. In that
sense, the dynamical $\pogr$-transformations generalise the gauge
transformations of classical dynamical $r$-matrices in
\cite{Etingof:1998aa}. We then apply these dynamical transformations to give a
complete classification of the classical dynamical $r$-matrices and Poisson
structures obtained from gauge fixing.  Section \ref{sec:outlook} contains the
outlook and our conclusions.

\section{Notations and conventions}
\label{sec:notation}

We denote by $\be_0=(1,0,0)$, $\be_1=(0,1,0)$, $\be_2=(0,0,1)$ the standard
basis of $\RR^3$ and use Einstein's summation convention throughout this
paper. Unless stated otherwise, all indices run from 0 to 2 and are raised and
lowered with the three-dimensional Minkowski metric $\eta=\diag(1,-1,-1)$. By
$\ee_{abc}$ we denote the totally antisymmetric tensor in three dimensions with
the convention $\ee_{012}=1$.  For vectors $\bx,\by\in\RR^3$, we use the
notation $\eta(\bx,\by)=\bx\cdot\by=\eta_{ab}x^ay^b$ and $\bx^2=\bx\cdot\bx$,
and we write $\bx\wedge\by$ for the vector with components
$(\bx\wedge\by)^a=\ee^{abc}x_by_c$. Note that this is a Lorentzian version of
the wedge product which does not coincide with the standard one.

We denote by $\logr \isoeq \PSL(2,\RR)$ the proper orthochronous Lorentz group
in three dimensions and by $\loal \isoeq \algebra{sl}(2,\RR)$ its Lie algebra.
In the following, we fix a set of generators $\{J_a\}_{a = 0, 1, 2}$ of $\loal$
such that the Lie bracket takes the form $[J_a, J_b] = \tensor{\ee}{_a_b^c}
J_c$.  As the representation of $\PSL(2,\RR)$ by $\logr$ matrices coincides
with its adjoint representation, we denote both representations by $\Ad$ in the
following:
\begin{equation*}
  g\cdot J_a\cdot g^{-1}=\tensor{\Ad(g)}{^b_a} J_b \qquad \forall g\in \logr.
\end{equation*}
The Poincaré group in three dimensions is the semidirect product $\pogr \equiv
\logr \ltimes \RR^3$ of the proper orthochronous Lorentz group $\logr$ and the
translation group $\RR^3$.  We parametrise elements of $\pogr$ as
\begin{equation*}
  (u,\ba)=(u,0)\cdot (\mathds{1},-\bj)=(u, -\Ad(u)\bj)\quad \text{with}\; u\in \logr,
    \bj, \ba\in\RR^3.
\end{equation*}
The group multiplication law then takes the form
\begin{equation*}
  (u_1,\ba_1)\cdot(u_2,\ba_2)
    =\bigl(u_1u_2,\ba_1+\Ad(u_1)\ba_2\bigr).
\end{equation*}
A basis of the Lie algebra $\poal$ is given by the basis $\{J_a\}_{a=0,1,2}$ of
$\loal$ together with a basis $\{P_a\}_{a=0,1,2}$ of the abelian Lie algebra
$\RR^3$.  In this basis, the Lie bracket takes the form
\begin{equation}\label{eq:poincare-algebra-bracket}
  [J_a,J_b]=\tensor{\ee}{_a_b^c} J_c,\qquad
  [J_a,P_b]=\tensor{\ee}{_a_b^c} P_c,\qquad
  [P_a,P_b]=0,
\end{equation}
and a non-degenerate $\Ad$-invariant symmetric bilinear form on $\poal$ is
given by
\begin{equation}\label{eq:pairing}
  \langle J_a,J_b\rangle=\langle P_a,P_b\rangle=0, \qquad
  \langle J_a,P_b\rangle=\eta_{ab}.
\end{equation}
All Cartan subalgebras of $\poal$ are abelian and can be parametrised in terms
of two vectors $\bx,\by\in\RR^3$ with $\bx^2\in\{1,-1\}$ and $\bx\cdot \by=0$
as
\begin{equation}\label{eq:csa}
  \csa = \Span\{x^aP_a, \, x^aJ_a + y^aP_a\}
\end{equation}

If the vector $\bx$ is timelike ($\bx^2=1$), then the associated Cartan
subalgebra $\csa$ is conjugate under the adjoint action of $\pogr$ to
$\Span\{P_0, J_0\}$.  If $\bx$ is spacelike ($\bx^2=-1$), then $\csa$ is
conjugate to $\Span\{P_1, J_1\}$. Note that the set \eqref{eq:csa} with a
lightlike vector $\bx\in\RR^3$ ($\bx^2=0$) does not form a Cartan subalgebra of
$\poal$ because it is not self-normalising.

In the following, we will also need the right- and left-invariant vector fields
on $\pogr$ associated with a basis $\{T_a\}_{a=0,\dots,5}$ of $\poal$. They are
given by
\begin{equation}\label{eq:vector-fields}
  L_a f(h)=\tdiffat{}{t}{t=0} f(e^{-tT_a} \cdot h),\quad
  R_a f(h)=\tdiffat{}{t}{t=0} f(h\cdot e^{tT_a})\qquad
  \forall f\in\cif(\pogr),
\end{equation}
where $e: \poal\to \pogr$, $\bx\mapsto e^{\bx}$ is the exponential map for
$\pogr$.  For the basis $\{J_a,P_a\}_{a=0,1,2}$, we denote by $P_a^{L}, P_a^R$,
respectively, the right- and left-invariant vector fields associated with the
translations and by $J_a^L, J_a^R$ the ones associated with the Lorentz
transformations. The former act trivially on functions that depend only on the
Lorentzian component of $\pogr$. For the latter, the action on such functions
coincides with the action of the right- and left-invariant vector fields of the
Lorentz group. The action of these vector fields on the coordinate functions
$j^a: \pogr\to\RR$, $(u,-\Ad(u)\bq)\mapsto q^a$ is given by
\begin{equation}\label{eq:action-of-vector-fields}
  \left.
  \begin{aligned}
    &P^L_a j^b (u,-\Ad(u)\bj)=\tensor{\Ad(u)}{_a^b}, & \quad
    &P^R_a j^b(u,-\Ad(u)\bj)=-\tensor{\delta}{_a^b}, \\
    &J^L_a j^b(u,-\Ad(u)\bj)=0, &
    &J^R_aj^b(u,-\Ad(u)\bj)=-\tensor{\ee}{^a^b_c}\,j^c.
  \end{aligned}
  \qquad\right\}
\end{equation}

\section{Physics background and motivation}
\label{sec:physics}

\subsection{Chern-Simons theory with gauge group \pogrforsection{} and the moduli space of flat \pogrforsection-connections}

\newcommand{\surf}{\ensuremath{S_{g,n}}}
\newcommand{\algorbit}{\mathcal{D}_i}
\newcommand{\grporbit}[1][i]{\mathcal{C}_#1}

In the following, we consider Chern-Simons theory with gauge group $\pogr$ on
manifolds of topology $M \approx \RR \times \surf$, where $\surf$ is an
oriented surface of genus $g$ with $n$ punctures.  In the absence of punctures,
the solutions of the theory are flat connections $A$ on an $\pogr$-principal
bundle over $M$.  The punctures of $\surf$ are incorporated
\cite{deSousaGerbert:1990aa} into the theory by assigning the coadjoint orbit
of an element $\algorbit \in \poal$ to the $i$-th puncture and coupling it
minimally to the connection $A$.  Parametrising the coadjoint orbit of
$\algorbit$ in terms of group-valued functions $h_i: \RR \to \pogr$, one
obtains the following expression for the Chern-Simons action:
\begin{equation}\label{eq:cs-action}
  S(A)=\int_M \langle A \wedge \diffd A + \tfrac23 A \wedge A \wedge A \rangle
         - 2 \sum_{i=1}^n \int_\RR \langle \algorbit, h_i^\inv A\big|_{l_i} h_i + h_i^\inv \diffd h_i \rangle \diffd t,
\end{equation}
where $\langle \;\;\rangle$ denotes the non-degenerate $\Ad$-invariant
symmetric bilinear form \eqref{eq:pairing} on $\poal$ and $l_i: \RR \to M$,
$i=1,\dots,n$, are the curves defined by the punctures.  Up to a topological
term, the Chern-Simons action is invariant under gauge transformations $\gamma
\in \cif(M, \pogr)$ that are constant along $l_i$: $A \mapsto \gamma^\inv A
\gamma + \gamma^{-1} \diffd\gamma$, $h_i \mapsto \gamma(l_i)^\inv h_i$.

The connections that extremise the action \eqref{eq:cs-action} are those that
are flat everywhere on $M$ except at $l_i$ ($i=1,\dots,n$), where their
curvature $F \equiv \diffd A + A \wedge A$ develops $\delta$-singularities.
From the Hamiltonian formulation of the theory one then obtains that the
gauge-invariant phase space of the theory is the moduli space $\phasespace$ of
flat $\pogr$-connections on $\surf$ modulo gauge transformations
\cite{Witten:1988aa,deSousaGerbert:1990aa}.

A convenient parametrisation of the moduli space is given by group
homomorphisms $h:\pi_1(\surf) \to \pogr$ that map the homotopy equivalence
class $m_i$ of a loop around the $i$-th puncture to the associated conjugacy
class
\begin{equation}\label{eq:particle-conjugacy-classes}
  \grporbit=\{h\cdot\exp(\algorbit)\cdot h^\inv \mid h\in \pogr\}.
\end{equation}
Two such group homomorphisms describe gauge-equivalent connections if and only
if they are related by conjugation with an element of $\pogr$. This implies
that the moduli space of flat $\pogr$-connections on $\surf$ is given by
\begin{align}\label{eq:phase-space}
  \phasespace
    &=\Hom_{\grporbit[1], \dots, \grporbit[n]}\big(\pi_1(\surf), \pogr\big)/\pogr\\
    &=\{h \in \Hom\bigl(\pi_1(\surf), \pogr\bigr) \mid h(m_i)\in\grporbit[i]\}/\pogr.\nonumber
\end{align}

\begin{figure}
  \centering
  \includegraphics{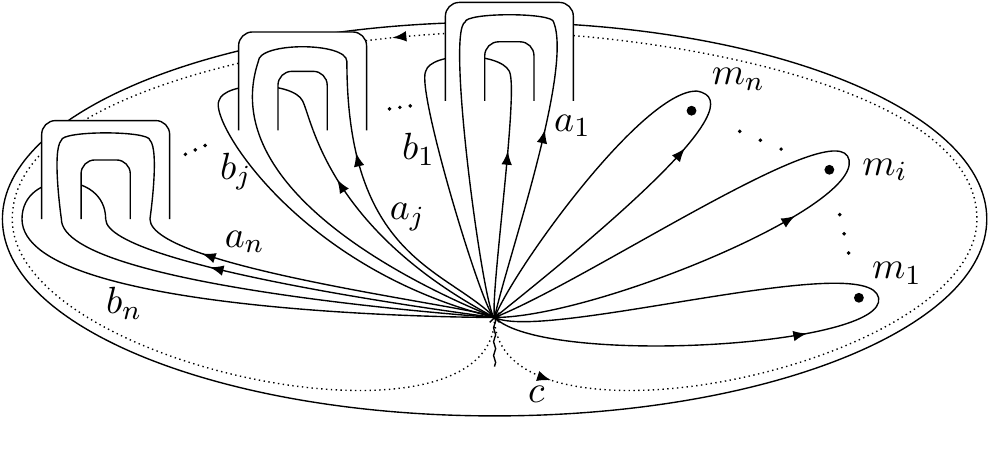}
  \caption{Generators of the fundamental group for an $n$-punctured genus $g$
    surface $\surf$.  The chosen generators of the fundamental group $\pi_1(\surf)$
    are the homotopy equivalence classes of the curves $m_1, \dots, m_n, a_1,
    b_1, \dots, a_n, b_n$.
    The short wavy line indicates the cilium that defines a linear ordering of
    the incident edges at the basepoint.}
  \label{fig:fundamental-group}
\end{figure}

The fundamental group $\pi_1(\surf)$ of an oriented genus $g$ surface $\surf$
with $n$ punctures is generated by the homotopy equivalence classes of a loop
$m_i$ ($i=1,\dots,n$) around each puncture and the $a$- and $b$-cycles
$a_j,b_j$ ($j=1,\dots,g$) of each handle as shown in Figure
\ref{fig:fundamental-group}. It has a single defining relation, which states
that the curve $c$ in Figure \ref{fig:fundamental-group} is contractible:
\begin{equation*}
  \pi_1(\surf)=\langle m_1,\dots,m_n, a_1,b_1,\dots,a_g,b_g \mid b_g a_g^\inv b_g^\inv a_g\cdots b_1 a_1^\inv b_1^\inv a_1 m_n\cdots m_1=1 \rangle.
\end{equation*}
By characterising the group homomorphisms in \eqref{eq:phase-space} in terms of
the images of the generators of $\pi_1(\surf)$, we can thus identify the moduli
space of flat connections with the set
\begin{multline}\label{eq:phase-space-with-holonomies}
  \phasespace = \{(M_1,\dots,M_n,A_1,B_1,\dots,A_g,B_g)\in\allpogr \mid \\
                  M_i\in\grporbit, \  [B_g,A_g^\inv] \cdot [B_1,A_1^\inv] \cdot M_n\cdots M_1=1\}/\pogr,
\end{multline}
where $[B_g, A_g^\inv]=B_g\cdot A_g^\inv\cdot B_g^\inv\cdot A_g$ denotes the
group commutator and the quotient is taken with respect to the diagonal action
of $\pogr$ on $\allpogr$.  In the gauge-theoretical description, the group
elements $M_1, \dots, B_g\in \pogr$ correspond to the path-ordered exponentials
of the gauge field $A$ along the closed curves $m_1, \dots, b_g$ displayed in
Figure \ref{fig:fundamental-group}.  In the following, we will sometimes refer
to these group elements as holonomies.

\subsection{Symplectic structure of the moduli space of flat connections}

The moduli space of flat $\pogr$-connections carries a canonical symplectic
structure \cite{Atiyah:1983aa} that is obtained via symplectic reduction from
the canonical symplectic structure on the space of connections on $\surf$.  A
convenient and explicit description of this symplectic structure is given in
the works of Alekseev and Malkin and Fock and Rosly \cite{Fock:1998aa,
  Alekseev:1995ab}. They describe the canonical symplectic structure on the
moduli space $\phasespace$ in terms of a (non-canonical) Poisson structure on
an enlarged ambient space $\extphasespace$. Via symplectic reduction, this
Poisson structure then induces the canonical symplectic structure on
$\phasespace$.

In the following, we will work with a specific form of the Poisson structure in
\cite{Fock:1998aa} which is associated with a choice of an ordered set of
generators of the fundamental group $\pi_1(\surf)$. It plays an important role
as a starting point for the combinatorial quantisation of the theory
\cite{Alekseev:1995aa,Alekseev:1996aa,Alekseev:1996ab,Buffenoir:1995aa}.  In
this description, the ambient space $\extphasespace$ is given by $n+2g$ copies
of the Poincar\'e group, $\extphasespace=\allpogr$, each corresponding to the
holonomy along a generator of the fundamental group $\pi_1(\surf)$.

Explicit expressions for this Poisson structure are given in Definition
\ref{def:fr} in the next section. Here, we only discuss its most important
structural features.  Firstly, the definition of the bracket requires a
classical $r$-matrix for the Lie algebra $\poal$, \ie an element
$r\in\poal\oo\poal$ that is a solution of the classical Yang-Baxter equation
\begin{equation*}
  [[r,r]] \defeq [r_{12}, r_{13}] + [r_{12}, r_{23}] + [r_{13}, r_{23}] = 0,
\end{equation*}
where $r_{12}=r^{\alpha\beta} T_\alpha\otimes T_\beta\otimes 1$,
$r_{13}=r^{\alpha\beta} T_\alpha\otimes 1\otimes T_\beta$,
$r_{23}=r^{\alpha\beta} 1\otimes T_\alpha\otimes T_\beta$.  This property is
needed to ensure that the bracket satisfies the Jacobi identity.  Moreover, it
is shown in \cite{Fock:1998aa} that this Poisson structure induces a symplectic
structure on the moduli space $\phasespace$, which agrees with its canonical
symplectic structure if and only if the symmetric part of $r$ is dual to the
pairing \eqref{eq:pairing} in the Chern-Simons action:
\begin{equation}\label{eq:r-symm}
  r_S \equiv r_{(s)}^{\alpha\beta} T_\alpha\otimes T_\beta
          = \tfrac 1 2 (P_a\otimes J^a+J^a\otimes P_a).
\end{equation}
Note that Fock and Rosly's Poisson structure on the ambient space
$\extphasespace$ is therefore non-canonical in two ways. Firstly, it depends on
the choice of a set of generators of the fundamental group $\pi_1(\surf)$ and
of a linear ordering of the incident edges at the basepoint. This ordering of
the edges is indicated in Figure \ref{fig:fundamental-group} by the short wavy
line (cilium) and gives rise to a partial ordering of the generators of
$\pi_1(\surf)$: $m_1 < \dots < m_n < a_1, b_1 < \dots < a_g, b_g$.

Secondly, the definition of the Poisson structure on the ambient space
$\extphasespace$ requires the choice of a classical $r$-matrix, which is
generally not unique; for a classification of the classical $r$-matrices for
$\poal$ see \cite{Stachura:1998aa}. However, it is shown in \cite{Fock:1998aa}
that the Poisson structures on $\allpogr$ associated with different choices of
generators and orderings of $\pi_1(\surf)$ and different choices of classical
$r$-matrices that satisfy \eqref{eq:r-symm} induce the same symplectic
structure on the moduli space $\phasespace$. This is apparent in formula
\eqref{eq:fr-bivector}, which shows that the Poisson bracket of two functions
on $\allpogr$ depends only on the symmetric component of $r$ if one of the two
functions is invariant under the diagonal action of $\pogr$ on $\allpogr$.

In the following, we work with the classical $r$-matrix that corresponds to the
structure of $\poal$ as a classical double of $\loal$. In terms of the basis
\eqref{eq:poincare-algebra-bracket} it is given by $ r = P_a \otimes J^a$.  It
is shown in \cite{Meusburger:2004aa}, see also the discussion in Section
\ref{subsec:dirac-intro}, that Fock and Rosly's Poisson structure for this case
can be formulated in terms of functions and certain vector fields on $\alllogr$
such that the Poisson bracket takes the form
\begin{equation*}
  \{f,g\}=0,\quad
  \{X, f\}=\mathcal L_Xf,\quad
  \{X,Y\}=[X,Y],
\end{equation*}
where $f,g\in\cif(\alllogr)$, $X,Y\in\text{Vec}(\alllogr)$, $\mathcal L_Xf$
denotes the Lie derivative and $[X,Y]$ the Lie bracket of vector fields on
$\alllogr$.

This implies that the associated Poisson algebra has a canonical $\NN$-grading,
in which the subspaces of homogeneous degree are given by homogeneous
polynomials in the vector fields $X$ with $\cif(\alllogr)$-valued
coefficients. This $\NN$-grading corresponds naturally to a physical dimension
of $\hbar$ and plays an important role in the quantisation of the theory.


\subsection{Quantisation}

The description of the moduli space of flat connections outlined in the
previous section serves as the starting point of the combinatorial quantisation
formalism \cite{Alekseev:1995aa, Alekseev:1996aa, Alekseev:1996ab,
  Buffenoir:1995aa} for Chern-Simons theories with compact, semisimple gauge
groups.  This formalism proceeds by quantising the auxiliary Poisson structure
on the ambient space $\extphasespace$ and then imposing the constraints in the
quantum theory. As this auxiliary Poisson structure is closely related to
certain Poisson structures from the theory of Poisson-Lie groups, the
corresponding quantum algebra is given in terms of quantum groups. The
implementation of the constraints in the quantum theory then reduces to a
problem from the representation theory of the associated quantum group, namely
to determining the invariant subspace in the tensor product of certain
representations of this quantum group.

While this formalism is well-established for Chern-Simons theories with
compact, semisimple gauge groups, for which the corresponding quantum groups
are universal enveloping algebras at a root of unity, it cannot be extended
straightforwardly to Chern-Simons theories with non-compact gauge groups.  This
is due to the fact that the representations of the corresponding quantum groups
no longer form a semisimple ribbon category and their characters become
distributions.

The generalisation of the combinatorial quantisation formalism to Chern-Simons
theory with gauge group $\SL(2,\CC)$ has been achieved in
\cite{Buffenoir:2002aa}, for partial results on semidirect product gauge groups
see \cite{Meusburger:2004aa, Meusburger:2010aa, Meusburger:2010bb}.  However,
there is currently no general quantisation formalism for moduli spaces of flat
connections associated with non-compact or non-semisimple groups.  Other
quantisation approaches such as Reshetikhin-Turaev invariants
\cite{Reshetikhin:1991aa} face similar problems.  Although there is work that
investigates their extension by analytic continuation \cite{Bar-Natan:1991aa,
  Witten:2011aa}, there is no general method that allows one to extend these
models to general non-compact or non-semisimple gauge groups.

For this reason, we pursue a different strategy, namely we implement the
constraints directly in the classical theory by means of the Dirac gauge fixing
procedure.  This potentially avoids the issues with the implementation of the
constraints in the quantisation of the theory and leads to an explicit
description of the canonical Poisson structure on the moduli space of flat
connections in terms of holonomies.  In the following, we apply this approach
to the description of the moduli space of flat $\pogr$-connections in terms of
the auxiliary Poisson structure on the ambient space $\extphasespace$.

A specific example of such a gauge fixing procedure is investigated in
\cite{Meusburger:2011aa}. It is shown there that in the application to
(2+1)-gravity, this gauge fixing procedure has a direct physical
interpretation. The gauge fixing conditions can be viewed as a prescription
that specifies an observer with respect to the geometry of the spacetime. The
resulting gauge-fixed Poisson structure then depends on two variables that
correspond to the total mass and internal angular momentum of the spacetime as
measured by this observer.

\subsection{The Dirac gauge fixing procedure}
\label{subsec:dirac-constraints}

In general, the Dirac gauge fixing procedure applies to constrained dynamical
systems, \ie to Poisson manifolds $(\extphasespace, \{,\})$ with a set of
constraint functions $\{\phi_i\}_{i=1,\dots,k}\subset\cif(\extphasespace)$.  In
the following, we restrict attention to the case where the constraints are
first-class and such that $0$ is a regular value of $\Phi=(\phi_1, \dots
\phi_k): \extphasespace \to \RR^k$.  Gauge fixing then amounts to imposing an
additional set of constraints
$\{\chi_j\}_{j=1,\dots,k}\subset\cif(\extphasespace)$, the gauge fixing
conditions, which must satisfy the following requirements
\cite{Henneaux:1994aa}:
\begin{enumerate}
\item It is possible to map any point $q \in \{p\in \extphasespace \mid
  \phi_i(p) = 0 \; \forall i=1,\dots,k\}$ to a point on the constraint surface
  $\csurface \defeq \{p\in \extphasespace \mid \phi_i(p) = 0 \text{ and }
  \chi_i(p)=0 \; \forall i=1,\dots,k\}$ with the flows that the first-class
  constraints $\phi_i$ generate via the Poisson bracket.
\item The matrix $C=(\{\chi_j, \phi_i\})_{i,j=1,\dots,k}$ is invertible
  everywhere on the constraint surface $\csurface$.
\end{enumerate}

The second condition implies that the gauge fixing conditions
$\{\chi_j\}_{j=1,\dots,k}$ together with the original constraints
$\{\phi_i\}_{i=1,\dots,k}$ can be collected in a single set
$\{C_i\}_{i=1,\dots,2k}$ of constraints such that $0$ is a regular value of the
function $C=(C_1,\dots,C_{2k}): M\to\RR^{2k}$ and for which the Dirac matrix
$D=(D_{ij})_{i,j=1,\dots,2k}$, $D_{ij} \defeq \{C_i, C_j\}$, is invertible
anywhere on the constraint surface $\csurface$.  The Dirac bracket of two
functions $f,g\in C^\infty(\csurface)$ is then defined by
\begin{equation*}
  \{f, g\}_D \defeq \{\tilde f, \tilde g\} - \smashoperator{\sum_{i,j=1}^{2k}} \{\tilde f, C_i\} (D^\inv)_{ij} \{C_j, \tilde g\},
\end{equation*}
where $\tilde f,\tilde g\in C^\infty(\extphasespace)$ are arbitrary extensions
of $f,g\in C^\infty(\csurface)$. The Dirac bracket does not depend on the
choice of the extensions $\tilde f, \tilde g$ and defines a Poisson structure
on $\csurface$ \cite{Dirac:1949aa,Dirac:1950aa,Henneaux:1994aa}.

The physical interpretation of this bracket can be summarised as follows: The
Poisson manifold $(\extphasespace,\{\,,\,\})$ plays the role of an extended,
non-gauge-invariant phase space which contains redundant degrees of freedom
corresponding to different descriptions of a single physical state. The
gauge-invariant or physical phase space is given as the quotient
$\phasespace=\equivalencequotientspace{\symplextphasespace}$, where
$\symplextphasespace\defeq\{p\in\extphasespace\,\vert\, \phi_i(p)=0\;\forall
i=1,\dots,k\}$ and two points on $\symplextphasespace$ are identified if they
are mapped into each other by the flows the first-class constraints $\phi_i$
generate via the the Poisson bracket. The associated equivalence classes are
called gauge orbits.

Imposing gauge fixing conditions amounts to selecting a representative in each
gauge orbit. The first requirement on the gauge fixing conditions ensures that
the gauge fixing conditions select at least one representative in every gauge
orbit. The second requirement ensures that they select at most one
representative in each gauge orbit.  Imposing gauge fixing conditions thus
amounts to constructing a diffeomorphism that identifies the quotient
$\phasespace=\equivalencequotientspace{\symplextphasespace}$ with the
submanifold $\csurface \subset \symplextphasespace \subset \extphasespace$.

From the viewpoint of symplectic reduction, this procedure can be interpreted
in the following way. If $(\extphasespace,\{\,\,\})$ is symplectic, then the
submanifold $\symplextphasespace\subset\extphasespace$ is coisotropic, and the
flows generated by the first-class constraints $\phi_i$ define a foliation of
$\symplextphasespace$.
If we think of $\symplextphasespace$ as a bundle over $\phasespace$, then
choosing a representative for each equivalence class in $\phasespace$ amounts
to specifying a global section on $\symplextphasespace$.  This is achieved via
the gauge fixing conditions, which define a diffeomorphism $\xi:
\csurface\subset\extphasespace\to\phasespace$, $p\mapsto [p]$.  The manifold
$\phasespace$ carries a canonical symplectic structure obtained via symplectic
reduction from the symplectic structure on $\extphasespace$.  The pull-back of
this symplectic structure to the submanifold $\csurface\subset\extphasespace$
with $\xi$ is the Dirac bracket on $\csurface$.

The situation is similar in the case where $(\extphasespace,\{\,\,\})$ is not
symplectic but there is a function $\Psi=(\psi_1,\dots,\psi_l)\in
C^\infty(\extphasespace,\mathcal \RR^l)$ such that $0$ is a regular value of
$\Psi$, $\Psi^\inv(0)$ a symplectic submanifold of $\extphasespace$ and
$\psi_1,\dots,\psi_l$ Poisson-commute with all functions on
$\extphasespace$. In this case, the reasoning above can be applied to the
submanifold $\Psi^\inv(0)$.

\subsection{Constraints and gauge fixing conditions for the moduli space}
\label{subsec:constraints-and-gauge-fixing}

The moduli space of flat $\pogr$-connections can be viewed as a constrained
system in the sense of Dirac.  From this viewpoint, the Poisson manifold
$(\extphasespace, \{\,,\})$ is identified with the ambient space
$\extphasespace=\allpogr$ equipped with Fock and Rosly's Poisson structure
\cite{Fock:1998aa}. From expression \eqref{eq:phase-space-with-holonomies} for
the moduli space of flat $\pogr$-connections, it is then apparent that the
moduli space is obtained from $\extphasespace$ by imposing a group-valued
constraint that arises from the defining relation of the fundamental group
$\pi_1(\surf)$, together with a set of constraints that restrict the holonomies
$M_1,\dots,M_n$ to the conjugacy classes \eqref{eq:particle-conjugacy-classes}.
The latter can be formulated as pairs of constraints of the form
\begin{equation*}
  \Tr(u_{M_i})-c_i \weaklyequal 0, \qquad
  \Tr(j_{M_i}^aJ_a\cdot u_{M_i})-d_i\weaklyequal 0,
\end{equation*}
with real parameters $c_i,d_i$ for each puncture. It turns out that these
constraints are Casimir functions of the Poisson bracket, \ie Poisson-commute
with all functions on $\extphasespace$. For this reason, reducing the Poisson
structure to the relevant conjugacy classes presents no difficulties and does
not require any gauge fixing.

The group-valued constraint from the defining relation of the fundamental group
$\pi_1(\surf)$ can be viewed as a set of six first-class constraints for the
Fock-Rosly bracket \cite{Meusburger:2011aa}.  Parametrising the holonomy of the
curve $c$ in Figure \ref{fig:fundamental-group} as
\begin{equation*}
  (u_C^\inv, \bj_C)
    \defeq M_1^\inv\cdots M_n^\inv [A_1^\inv, B_1]\cdots[A_g^\inv, B_g],
\end{equation*}
one can express this group-valued constraint in the form of the six constraints
\begin{equation}\label{eq:constraints-in-jp}
  \Tr(J_a\cdot u_C) \weaklyequal 0,\qquad  j_C^a\weaklyequal 0\qquad\forall a\in\{0,1,2\}.
\end{equation}
The Poisson brackets of these six constraint functions are closely related to
the Lie bracket of $\poal$. The associated gauge transformations which they
generate via the Poisson bracket are given by the diagonal action of $\pogr$ on
$\allpogr.$

A specific choice of gauge fixing conditions for the constraints
\eqref{eq:constraints-in-jp} is investigated in \cite{Meusburger:2011aa}. It
imposes gauge fixing conditions on the holonomies associated with two punctures
on the surface $\surf$ and derives the associated Dirac bracket. In this paper
we consider general gauge fixing conditions, subject to certain structural
requirements, and investigate the resulting Dirac brackets. We require that the
gauge fixing conditions satisfy the conditions 1 and 2 in Section
\ref{subsec:dirac-constraints} and are subject to the following two additional
restrictions:
\begin{enumerate}[\quad a.]
\item The gauge fixing conditions depend only on the holonomies $M_i$, $M_j$
  associated with two punctures on the surface $\surf$.
\item The gauge fixing conditions depend at most linearly on the variables
  $j_{M_i}$ and $j_{M_j}$ associated with these holonomies.
\end{enumerate}

The first condition is motivated by convenience and by physics
considerations. Although it is feasible in principle to impose gauge fixing
conditions that involve the holonomies of more than two punctures, this would
complicate many details of the description without adding much on the
conceptual level. Moreover, in the application to (2+1)-gravity, gauge fixing
conditions based on the holonomies of two punctures have a direct physical
interpretation, while the interpretation of a complicated gauge fixing
condition involving more than two punctures is less obvious.

Note also that the first condition allows us to restrict attention to gauge
fixing conditions that depend only on the holonomies $M_1$, $M_2$ of the first
two punctures on $\surf$.  This is due to the fact that different orderings of
the punctures are related by the action of the braid group on $\surf$ and the
braid group on the associated surface $\surf\setminus D$ with a disc removed
\cite{Birman:1974aa}. It is shown in \cite{Meusburger:2005aa} that the braid
group of the surface $\surf\setminus D$ acts by Poisson isomorphisms on the
Poisson manifold $(\extphasespace=\allpogr,\{\,,\,\})$. The action of the braid
group thus allows one to permute the punctures and to suppose that the gauge
fixing conditions depend only on the holonomies of the first two punctures.

The second condition is motivated by structural considerations, namely the wish
to preserve the natural $\NN$-grading of the Poisson structure. As we will see
in the following, gauge fixing conditions that are non-linear in the variables
$j_{M_1}$ or $j_{M_2}$ and, consequently, non-linear in the vector fields on
$\alllogr$, would compromise this grading.  However, the grading is an
important structural feature of the theory and plays a central role in its
quantisation \cite{Meusburger:2004aa,Meusburger:2010aa,Meusburger:2010bb}. For
this reason, it seems natural to impose that it is preserved by the gauge
fixing procedure.

Conditions 1 and 2 from Section \ref{subsec:dirac-constraints} together with
the additional assumptions a and b above imply that the gauge fixing conditions
can be brought into the form
\begin{equation*}
  \sum_{i=1}^2 \ThetaFam{1}{M_i}{a}\,j^a_{M_i} \weaklyequal 0, \quad
  \sum_{i=1}^2 \ThetaFam{2}{M_i}{a}\,j^a_{M_i} \weaklyequal 0, \quad
  \sum_{i=1}^2 \ThetaFam{3}{M_i}{a}\,j^a_{M_i} \weaklyequal 0, \quad
  \Delta_1 \weaklyequal 0, \quad
  \Delta_2 \weaklyequal 0, \quad
  \Delta_3 \weaklyequal 0,
\end{equation*}
where $\ThetaFam{j}{M_i}{a}, \Delta_j \in \cif(\logr\times \logr)$ and the two
copies of the Lorentz group $\logr$ are identified with the Lorentzian
components of the holonomies $M_1$ and $M_2$.  These gauge fixing conditions
allow one to express the two constrained holonomies $M_1$ and $M_2$ as
functions of the four fixed parameters that characterise the conjugacy classes
$\grporbit[1],\grporbit[2]$ and of two conjugation-invariant dynamical
variables $\psi,\alpha$, which depend only on the product $M_2\cdot M_1$.  As
there are many possible definitions of these variables, we will not adhere to
one of them, but impose that they are given in terms of the Lorentzian and
translational components of the product $M_2\cdot M_1=(u_{12},-\Ad(u_{12})
\bj_{12})$ as
\begin{equation*}
  \psi=f(\Tr(u_{12})), \quad
  \alpha=g(\Tr(u_{12})) \Tr(j_{12}^aJ_a\cdot u_{12})+h(\Tr(u_{12})),
\end{equation*}
with diffeomorphisms $f,g\in\cif(\RR)$ and a smooth function $h\in\cif(\RR)$.

As for the gauge fixing conditions investigated in \cite{Meusburger:2011aa},
the Dirac gauge fixing procedure with these gauge fixing conditions gives rise
to a Poisson structure $\{\,,\}_D$ on the constraint surface $\csurface \subset
\RR^2\times\restpogr$, where $\RR^2$ is parametrised by the variables
$\psi,\alpha$ and $\restpogr$ by the non-gauge-fixed holonomies
$M_3,\dots,B_g$.  This Poisson structure is derived in the next section.

\section{Gauge fixing and solutions of the classical dynamical Yang-Baxter equation}
\label{sec:gaugefixing}

\subsection{General form of the Dirac bracket}
\label{subsec:dirac-intro}

In this section, we derive explicitly the Dirac bracket obtained by gauge
fixing the auxiliary Poisson structure from \cite{Fock:1998aa} on the ambient
space $\allpogr$. While the Poisson structure in \cite{Fock:1998aa} is
associated with general ciliated fat graphs on a genus $g$-surface $S_{g,n}$
with $n$ punctures, we restrict attention to the case where the graph is a set
of generators of the fundamental group $\pi_1(S_{g,n})$ as depicted in Figure
\ref{fig:fundamental-group}. In that case, the Poisson structure from
\cite{Fock:1998aa} takes the following form.

\begin{definition}[\cite{Fock:1998aa}]\label{def:fr}
Let $G$ be a Lie group with Lie algebra $\ag$,
$\{T_\alpha\}_{\alpha=1,\dots,\text{dim}(G)}$ a basis of $\ag$ and
$r=r^{\alpha\beta}T_\alpha\oo T_\beta\in\ag\otimes\ag$.  Then Fock and Rosly's
bivector $\frbivector$ is the antisymmetric section of the bundle
$TG^{n+2g}\otimes TG^{n+2g}$ defined by
\begin{align}\label{eq:fr-bivector}
    &\frbivector
    = \tfrac 1 2 r_{(a)}^{\alpha\beta} \Bigl(\smashoperator{\sum_{i=1}^{n+2g}} L^{i}_\alpha+R^{i}_\alpha\Bigr)\oo\Bigl(\smashoperator{\sum_{j=1}^{n+2g}} L^{j}_\beta+R^{j}_\beta\Bigr)
    + \tfrac 1 2 r_{(s)}^{\alpha\beta} \smashoperator{\sum_{1\leq i<j\leq n+2g}} \left(L^{i}_\alpha+R^{i}_\alpha\right)\wedge\left(L^{j}_\beta+R^{j}_\beta\right)  \\[-.5em]
    &\!+\! \tfrac 1 2 r_{(s)}^{\alpha\beta}\Big(
      \sum_{i=1}^n R^{i}_\alpha \!\wedge\! R^{i}_\beta \!+\!
      \sum_{j=1}^{g} \Big[
        L^{n+2j}_\alpha \!\wedge\! L^{n+2j}_\beta \!-\!
        (R^{n+2j-1}_\alpha \!+\! L^{n+2j-1}_\alpha) \!\wedge\! R^{n+2j}_\beta \!-\!
        L^{n+2j-1}_\alpha \!\wedge\! L^{n+2j}_\beta
      \Big]
    \Big),\nonumber
\end{align}
where $L_\alpha^i$ and $R_\alpha^i$ denote the right- and left-invariant vector
fields associated with the different components of $G^{n+2g}$ and the basis
elements $T_\alpha$:
\begin{align*}
  L^i_\alpha f(g_1,\dots,g_{n+2g})&=\tdiffat{}{t}{t=0} f(u_1,\dots,u_{i-1}, e^{-t T_\alpha}\cdot u_i, u_{i+1},\dots,u_{n+2g}),\\
  R^i_\alpha f(g_1,\dots,g_{n+2g})&=\tdiffat{}{t}{t=0} f(u_1,\dots,u_{i-1}, u_i\cdot e^{t T_\alpha}, u_{i+1},\dots,u_{n+2g}),
\end{align*}
and $r^{\alpha\beta}_{(a)}$, $r^{\alpha\beta}_{(s)}$ the coefficients of the
antisymmetric and symmetric component of $r$:
\begin{equation*}
  r^{\alpha\beta}_{(a)}=\tfrac 1 2 (r^{\alpha\beta}-r^{\beta\alpha})\qquad r^{\alpha\beta}_{(s)}=\tfrac 1 2 (r^{\alpha\beta}+r^{\beta\alpha}).
\end{equation*}
The associated bracket on $G^{n+2g}$ is the antisymmetric bilinear map $\{\;\}:
C^\infty(G^{n+2g})\times C^\infty(G^{n+2g})\to C^\infty(G^{n+2g})$ given by
$$\{f,g\}=\frbivector(\diffd f \otimes \diffd g)\qquad\forall f,g\in C^\infty(G^{n+2g}).$$
\end{definition}

The main advantage of this description is that it defines an auxiliary Poisson
structure on $G^{n+2g}$ that is given in terms of a classical $r$-matrix for
$\ag$ and induces the canonical symplectic structure on the moduli space of
flat $G$-connections on $S_{g,n}$.

\begin{theorem}[\cite{Fock:1998aa}]
If $r$ is a solution of the classical Yang-Baxter equation
\begin{equation*}
  [[r,r]]=[r_{12}, r_{13}]+[r_{12}, r_{23}]+[r_{13}, r_{23}]=0,
\end{equation*}
then $\frbivector$ defines a Poisson structure on $G^{n+2g}$. If additionally
$\langle\;,\;\rangle$ is a non-degenerate $\Ad$-invariant symmetric bilinear
form on $\ag$ and
\begin{equation*}
  r^{\alpha\beta}_{(s)}=\tfrac12\kappa^{\alpha\beta}
  \quad\text{ with }
  \kappa^{\alpha\beta}\kappa_{\beta\gamma}=\tensor{\delta}{^\alpha_\gamma},\;
  \kappa_{\alpha\beta}=\langle T_\alpha, T_\beta\rangle,
\end{equation*}
then the Poisson structure defined by $\frbivector$ induces the canonical
symplectic structure on the moduli space of flat $G$-connections on $S_{g,n}$.
\end{theorem}

It is shown in \cite{Alekseev:1995ab} that this Poisson structure can be
identified with the direct product of $n$ copies of the dual Poisson-Lie
structure on $G$ and $g$ copies of the Heisenberg double Poisson structure
associated with $G$ and $r$.

For the case at hand, where the Lie group is the Poincar\'e group in three
dimensions, $G=\pogr$, a classical $r$-matrix is given by $r=P_a\oo J^a$. The
corresponding Poisson structure on $\allpogr$ is computed explicitly in
\cite{Meusburger:2003aa}, and in \cite{Meusburger:2004aa} it is shown that this
Poisson structure can be identified with the direct product of $n$ copies of
the dual Poisson-Lie structure on $\pogr$ and $2g$ copies of the cotangent
bundle Poisson structure:
\begin{equation*}
  (\allpogr,\{\})=\underbrace{\pogr^*\times\ldots\times \pogr^*}_{n\,\times}\times \underbrace{T^*(\logr)\times\ldots\times T^*(\logr)}_{2g\,\times}.
\end{equation*}
From this description, it is directly apparent that the symplectic leaves of
this Poisson structure are of the form
$\grporbit[1]\times\dots\times\grporbit[n]\times\pogr^{2g}$, where
$\grporbit[i]\subset \pogr$ are fixed conjugacy classes.

In the following, we will not use this identification, but we will work with a
description of the Poisson structure that is closer to the formula in
Definition \ref{def:fr}. This formulation has the advantage that it is more
adapted to physics applications, especially in the Chern-Simons formulation of
$(2+1)$-gravity, and that its geometrical interpretation is more apparent.  To
emphasise the geometrical interpretation of the variables and their relation
with the set of generators of the fundamental group $\pi_1(S_{g,n})$ in Figure
\ref{fig:fundamental-group} we denote elements of $\allpogr$ and $\alllogr$ as,
respectively,
\begin{equation*}
  (M_1,\dots,M_n,A_1,B_1,\dots,A_g,B_g)\in \allpogr,\qquad
  (u_{M_1},\dots,u_{B_g})\in\alllogr,
\end{equation*}
and write $J_{X}^{L,a}$, $J_Y^{R,a}$ for the associated right- and
left-invariant vector fields on $\alllogr$:
\begin{align*}
  J_{X}^{L,a}f(u_{M_1},\dots,u_{B_g})&=\tdiffat{}{t}{t=0} f(u_{M_1},\dots, e^{-tJ_a}\cdot u_X,\dots,u_{B_g}),\\
  J_{X}^{R,a}f(u_{M_1},\dots,u_{B_g})&=\tdiffat{}{t}{t=0} f(u_{M_1},\dots, u_X\cdot e^{tJ_a},\dots,u_{B_g}).
\end{align*}

It is shown in \cite{Meusburger:2004aa} that by using the identification
$\pogr=T\logr$ and the classical $r$-matrix $r=P_a\oo J^a$, the Poisson
structure given by \eqref{eq:fr-bivector} can be expressed in terms of
functions $f\in C^\infty(\alllogr)$ and certain vector fields on
$\alllogr$. For this, one identifies the coordinate functions $j^a_X:
(M_1,\dots,B_g)\mapsto q^a_X$, where we use the parametrisation
$X=(u_X,-\Ad(u_X)\bq)$ for $X\in\{M_1,\dots,B_g\}$, with certain vector fields
on $\alllogr$.  The Poisson bracket of two variables $j^a_X,j^b_Y$ then
coincides with the Lie bracket of the associated vector fields, the Poisson
bracket of a variable $j^a_X$ with a function on $\alllogr$ coincides with its
Lie derivative and any two functions on $\alllogr$ Poisson-commute.

\begin{theorem}[\cite{Meusburger:2003aa, Meusburger:2004aa}]\label{thm:fr-simple}
For $G=\pogr$ and $r=P_a\oo J^a$, the Poisson structure \eqref{eq:fr-bivector}
on $\allpogr=T\alllogr$ is characterised uniquely in terms of vector fields
$j^a_X$, $a\in\{0,1,2\}$, $X\in\{M_1,\dots,B_g\}$ on $\alllogr$ and functions
on $\alllogr$:
\begin{equation*}
  \{f,g\}=0,\qquad
  \{j^a_X, f\}=\mathcal L_{j^a_X} f,\qquad
  \{j^a_X,j^b_Y\}=[j^a_{X}, j^b_{Y}]\qquad
  \forall f\in C^\infty(\alllogr),
\end{equation*}
where $\mathcal L$ denotes the Lie derivative and $[\;,\;]$ the Lie bracket on
$\alllogr$.  The vector fields $j^a_X$ on $\alllogr$ are given by:
\begin{align*}
  &j^a_{M_i}=-\left(J_{M_i}^{R,a}+J_{M_i}^{L,a}\right)-\tensor{(\mathds{1}\!-\!\Ad(u_{M_i}))}{^a_b} \smashoperator{\sum_{Y>M_i}} \left(J^{R,b}_Y\!+\!J^{L,b}_Y\right),\\
  &j^a_{A_j}=-\left(J_{A_j}^{R,a}+J_{A_j}^{L,a}+J_{B_j}^{L,a}+\tensor{(\mathds{1}\!-\!\Ad(u_{A_j}^\inv u_{B_j}))}{^a_b} J_{B_j}^{R,b}\right)-\tensor{(\mathds{1}\!-\!\Ad(u_{A_j}))}{^a_b} \smashoperator{\sum_{Y>A_j}} \left(J^{R,b}_Y\!+\!J^{L,b}_Y\right),\\
  &j^a_{B_j}=-\left(J_{B_j}^{R,a}+J_{B_j}^{L,a}+J_{A_j}^{L,b}\right)-\tensor{(\mathds{1}\!-\!\Ad(u_{B_j}))}{^a_b} \smashoperator{\sum_{Y>A_j}} \left(J^{R,b}_Y\!+\!J^{L,b}_Y\right),
\end{align*}
where $Y>X$ refers to the partial ordering of the generators of $\pi_1(\surf)$:
$Y>X$ if $X=M_i$ and $Y=M_j$ with $i<j$ or if $X\in\{A_i,B_i\}$,
$Y\in\{A_j,B_j\}$ with $i<j$ or if $X\in\{M_1,\dots,M_n\}$,
$Y\in\{A_1,B_1,\dots,A_g,B_g\}$.
\end{theorem}

We will now determine the Dirac bracket associated with the Poisson structure
on $\allpogr$ and certain smooth constraint functions on $\allpogr$.  The Dirac
bracket is a well-established formalism from the theory of constrained
Hamiltonian systems and, in a certain sense, can be viewed as the Poisson
counterpart or the physicist's version of symplectic reduction. A more detailed
discussion of this is given in Section \ref{subsec:dirac-constraints}, for the
general theory we refer the reader to
\cite{Dirac:1949aa,Dirac:1950aa,Henneaux:1994aa}.

\begin{definition}[\cite{Dirac:1949aa,Dirac:1950aa}]\label{def:dirac-bracket}
Let $(M,\{\,,\,\})$ be an $n$-dimensional Poisson manifold, $k<n$, and
$C=(C_1,\dots,C_k): M\to\RR^k$ a smooth function such that $0$ is a regular
value of $C$ and the matrix $D(p)=(\{C_i,C_j\}(p))_{i,j=1,\dots,k}$ is
invertible for all $p\in \csurface=C^\inv(0)$. The \definee{Dirac bracket} for
$C$ is the antisymmetric bilinear map $\{\,,\,\}_D: C^\infty(\csurface)\times
C^\infty(\csurface)\to C^\infty(\csurface)$ defined by
\begin{equation*}
  \{f,g\}_D=\{\tilde f,\tilde g\}\oncsurface - \smashoperator{\sum_{i,j=1}^k} \{\tilde f, C_i\} \cdot (D\oncsurface)^\inv_{ij} \cdot \{C_j,\tilde g\}\oncsurface,
\end{equation*}
where $\tilde f,\tilde g\in C^\infty(M)$ are arbitrary extensions of $f, g\in
C^\infty(\csurface)$: $\tilde f\oncsurface=f$, $\tilde g\oncsurface=g$. The
Dirac bracket does not depend on the choice of these extensions and defines a
Poisson structure on $\csurface$.  The submanifold $\csurface=C^\inv(0)\subset
M$ is called \definee{constraint surface}, the functions $C_i: M\to\RR$ are
called \definee{constraint functions}.
\end{definition}

The aim is now to determine the Dirac bracket for the Poisson structure from
Definition~\ref{def:fr} for constraint functions that relate this Poisson
structure to the moduli space of flat $\pogr$-connections:
\begin{align}\label{eq:math-phase-space}
  \phasespace
  &=\{h\in \Hom\bigl(\pi_1(S_{g,n}), \pogr\bigr) \mid h(m_i)\in\grporbit[i]\}/\pogr\\
  &\isoeq\{(M_1, \dots, B_g)\in\grporbit[1]\times\dots\times\grporbit[n]\times\pogr^{2g} \mid \nonumber\\
  &\hspace{9em} [B_g,A_g^\inv]\cdots[B_1,A_1^\inv]\cdot M_n\cdots M_1=1\}/\pogr.\nonumber
\end{align}
As discussed in Section \ref{subsec:constraints-and-gauge-fixing}, this leads
to constraint functions of the form
\begin{equation}\label{eq:generic-constraints}
  \left.
  \begin{aligned}
    &C_1= j_C^0, &\quad &C_2= j_C^1, &\quad &C_3= j_C^2, \\
    &C_4=\textstyle{\sum_{i=1}^2} \ThetaFam{1}{M_i}{a}\,j^a_{M_i}, &\quad
    &C_5=\textstyle{\sum_{i=1}^2} \ThetaFam{2}{M_i}{a}\,j^a_{M_i}, &\quad
    &C_6=\textstyle{\sum_{i=1}^2} \ThetaFam{3}{M_i}{a}\,j^a_{M_i}, \\
    &C_7=\Tr(J_0\cdot u_C), &\quad &C_8=\Tr(J_1\cdot u_C), &\quad &C_9=\Tr(J_2\cdot u_C), \\
    &C_{10}=\Delta_1, &\quad
    &C_{11}=\Delta_2, &\quad
    &C_{12}=\Delta_3,
  \end{aligned}
  \quad\right\}
\end{equation}
where $j^a_{M_i}$ is defined as in Theorem \ref{thm:fr-simple} and
$\ThetaFam{j}{M_i}{a}, \Delta_j \in \cif(\logr\times \logr)$ are functions that
depend only on the $\logr$-part of the first two copies of $\pogr$ and
\begin{equation*}
  (u_C^\inv, \bj_C)
    \defeq M_1^\inv\cdots M_n^\inv \cdot [A_1^\inv, B_1]\cdots[A_g^\inv, B_g].
\end{equation*}
The functions $(C_i)_{i=1,2,3,7,8,9}$ have an interpretation as first-class
constraints in the Dirac gauge fixing formalism and the functions
$(C_i)_{i=4,5,6,10,11,12}$ play the role of gauge fixing conditions. While the
former are fixed and implement the condition
$[B_g,A_g^\inv]\cdots[B_1,A_1^\inv]\cdot M_n\cdots M_1=1$, the latter involve
functions $\ThetaFam{j}{M_i}{a}, \Delta_j $ which can be chosen arbitrarily as
long as the requirements from Definition \ref{def:dirac-bracket} and the
conditions a and b from Section~\ref{subsec:constraints-and-gauge-fixing} are
met. The latter ensure that the constraint functions are adapted to the tangent
bundle structure of $\pogr=T\logr$. Different choices of these functions
correspond to different gauge choices. They implement the quotient by $\pogr$
in \eqref{eq:math-phase-space} and restrict the variables $M_1,M_2$ in such a
way that for all points $(M_1,\dots,B_g)\in \csurface=C^\inv(0)$, the
components $M_1,M_2\in \pogr$ are determined uniquely by two real parameters
\begin{align}\label{eq:psi-alpha-general}
  \psi=f(\Tr(u_{12})), \quad
  \alpha=g(\Tr(u_{12})) \Tr(j_{12}^aJ_a\cdot u_{12})+h(\Tr(u_{12})),
\end{align}
where $f,g\in\cif(\RR)$ are arbitrary diffeomorphisms and $h\in\cif(\RR)$.
This allows us to identify the constraint surface $\csurface=C^\inv(0)$ with a
subset of $\RR^2\times\restpogr$, where the $\RR^2$ is parametrised by
$(\psi,\alpha)$ and $\restpogr$ by $(M_3,\dots,B_g)$.

Given the expressions for the Poisson structure on $\allpogr$ and the
constraint functions \eqref{eq:generic-constraints}, we can explicitly compute
the associated Dirac bracket and obtain a Poisson structure on
$\csurface\subset\RR^2\times\restpogr$.

\begin{theorem}\label{thm:generic-dirac-bracket}
  For all constraint functions of the form \eqref{eq:generic-constraints} that
  satisfy the requirements in Definition \ref{def:dirac-bracket} and conditions
  a and b from Section \ref{subsec:constraints-and-gauge-fixing}, the
  associated Dirac bracket defines a Poisson structure $\{\,,\,\}_D$ on
  $\csurface \subset \RR^2\times\restpogr$, which takes the following form:
  \begin{enumerate}
  \item The Dirac bracket of $\psi$ and $\alpha$ vanishes: $\{\psi, \alpha\}_D
    = 0$.
  \item For all $X \in \{M_3, \dots, B_g\}$ and $f\in\cif(\restlogr)$:
    \begin{equation*}
      \begin{aligned}
        \{\psi, f \}_D &= 0, &
        \{\psi, \bj_X \}_D &=
          -\idadi{X} \, \bq_\psi, \\
        \{\alpha, f\}_D &=
          \smashoperator{\sum_{Y\in\{M_3,\dots,B_g\}}} q_\alpha^a(J_a^{R,Y}+J_a^{L,Y})f, &
        \{\alpha, \bj_X\}_D &=
            -\idadi{X} \bq_\theta - \bq_\alpha \wedge \bj_X,
      \end{aligned}
    \end{equation*}
    with $\bq_\psi, \bq_\alpha,\bq_\theta: \RR^2 \to \RR^3$ satisfying
    $\bq_\psi \wedge \bq_\alpha = 0$ and
    $\partial_\alpha\bq_\psi=\partial_\alpha\bq_\alpha=\partial_\alpha^2
    \bq_\theta= 0$.
  \item For $F, G \in \cif(\restpogr)$: $$\{F, G\}_D = \frrestbivector(\diffd F
    \otimes \diffd G),$$ where $\frrestbivector$ is the Poisson bivector
    \eqref{eq:fr-bivector} and $r: \RR^2 \to \poal \oo \poal$ is given by
    \begin{equation*}
      r (\psi,\alpha)= P_a \otimes J^a - V^{bc}(\psi)(P_b \otimes J_c - J_c \otimes P_b) + \ee^{bcd}m_d(\psi,\alpha) P_b \otimes P_c,
    \end{equation*}
    where $V:\RR \to \Mat(3, \RR)$ and $\bm:\RR^2 \to \RR^3$ satisfies
    $\partial_\alpha^2 \bm=0$.
  \end{enumerate}
\end{theorem}

\begin{proof}
  The proof is a direct generalisation of the proof of Theorem 5.1 in
  \cite{Meusburger:2011aa}.
  \begin{enumerate}
  \item The Dirac matrix associated to the constraints
    \eqref{eq:generic-constraints} takes the form
    \begin{equation*}
      D = \begin{pmatrix}J & P \\ -P^T & 0\end{pmatrix} \quad\text{with}\quad
      J \defeq (\{C_i, C_j\})_{i,j=1,\ldots,6}, \;
      P \defeq (\{C_i,C_{j+6}\})_{i,j=1,\ldots,6}.
    \end{equation*}
    On the constraint surface, the $(6 \times 6)$-matrices $J$ and $P$ can be
    expressed as
    \begin{equation*}
      J\oncsurface=\begin{pmatrix}0 & H \\ -H^T & G\end{pmatrix}, \qquad
      P\oncsurface=\begin{pmatrix}0 & A \\ B & C\end{pmatrix},
    \end{equation*}
    with $(3\times 3)$-matrices $A, B, C, G, H$ given by
    \begin{equation*}
      \left.
      \begin{gathered}
        A_{ij} \defeq \{C_i, C_{j+9}\}\oncsurface, \quad
        B_{ij} \defeq \{C_{i+3}, C_{j+6}\}\oncsurface, \quad
        C_{ij} \defeq \{C_{i+3}, C_{j+9}\}\oncsurface, \\
        G_{ij} \defeq \{C_{i+3}, C_{j+3}\}\oncsurface, \quad
        H_{ij} \defeq \{C_i, C_{j+3}\}\oncsurface,
      \end{gathered}
      \;\right\}\; i,j=1,2,3.
    \end{equation*}
    This implies that the inverse of the Dirac matrix $D$ on the constraint
    surface is given by
    \begin{align*}
      &(D\oncsurface)^\inv = \begin{pmatrix}0 & -(P^\inv)^T \\ (P\oncsurface)^\inv & (P\oncsurface)^\inv (J\oncsurface) (P\oncsurface^\inv)^T\end{pmatrix}
      \;\;\text{with}\;\;
      (P\oncsurface)^\inv = \begin{pmatrix}-B^\inv C A^\inv & B^\inv \\ A^\inv & 0\end{pmatrix}, \\
      &(P\oncsurface)^\inv (J\oncsurface) (P\oncsurface^\inv)^T =
          \begin{pmatrix}
            B^\inv\bigl[G-CA^\inv H+(CA^\inv H)^T\bigr](B^\inv)^T & -B^\inv H^T(A^\inv)^T \\
            A^\inv H (B^\inv)^T & 0
          \end{pmatrix}.
    \end{align*}
    Inserting these expression into the general formula in Definition
    \ref{def:dirac-bracket}, one finds that for all $X, Y \in \{M_1, \dots,
    B_g\}$ and $f,g\in\cif(\alllogr)$, the Dirac bracket takes the form
    \begin{subequations}\label{eq:generic-dirac-general}
    \begin{align}
      \{f, g\}_D &= 0, \\
      \begin{split}
        \{j_X^a, f\}_D &= \{j_X^a, f\}\oncsurface + \smashoperator{\sum_{i,j=1}^3} \Bigl[
          \{j_X^a, C_{i+6}\} (B^\inv)_{ij} \{f, C_{j+3}\}
          +\{j_X^a, C_{i+9}\} (A^\inv)_{ij} \{f, C_j\} \\[-0.6em] &\hspace{8em}
          -\{j_X^a, C_{i+6}\} (B^\inv C A^\inv)_{ij} \{f, C_j\}
        \Bigr]\oncsurface,
      \end{split} \raisetag{3.8em}\\
      \begin{split}
        \{j^a_X, j^b_Y\}_D &= \{j^a_X,j^b_Y\}\oncsurface
                        + \sum_{i=1}^6\sum_{j=7}^{12} \{j^a_X, C_i\}\oncsurface (D\oncsurface^\inv)_{ij} \{j^b_Y, C_j\}\oncsurface \\[-0.8em]
          + &\sum_{i=7}^{12}\sum_{j=1}^6 \{j^a_X, C_i\}\oncsurface (D\oncsurface^\inv)_{ij} \{j^b_Y, C_j\}\oncsurface
          + \sum_{i=7}^{12}\sum_{j=7}^{12} \{j^a_X, C_i\}\oncsurface (D\oncsurface^\inv)_{ij} \{j^b_Y, C_j\}\oncsurface.
      \end{split}\raisetag{5.0em}
    \end{align}
    \end{subequations}

  \item To prove the relations for the brackets involving $\psi$ and $\alpha$,
    we use \eqref{eq:generic-dirac-general} to compute the Dirac brackets of
    $j_{M_1}$, $j_{M_2}$ and functions $g\in\cif(\logr\times \logr)$ of the
    variables $u_{M_1}$, $u_{M_2}$ with functions $f\in\cif(\restlogr)$ of the
    variables $M_3,\dots,B_g$.  It follows directly from the block form of
    $(D\oncsurface)^\inv$ that $\{g, f\}_D = 0$ and hence $\{\psi, f\}_D = 0$
    by \eqref{eq:psi-alpha-general}.  For $X\in\{M_3,\dots,B_g\}$, we have
    $\{j^a_X, C_{i+9}\} = 0$ for all $i\in\{1,2,3\}$ and thus
    \begin{equation*}
      \{j_{X}^a, g\}_D = -\tidadi{X}{^a_c} \smashoperator{\sum_{j,k=1}^3} \tensor{W}{^c_{k-1}} \bigl[
          (B^{-1})_{kj} \{g, C_{j+3}\} - (B^{-1} C A^{-1})_{kj} \{g, C_j\}
        \bigr]\oncsurface,
    \end{equation*}
    where $W: \RR \to \Mat(3, \RR)$ is a function of the variable $\psi$ from
    \eqref{eq:psi-alpha-general} defined by the condition $\{j_X^a, C_{i+6}\} =
    -\tidadi{X}{^a_b} \tensor{W}{^b_{i-1}}$ for all $i = 1, 2, 3$,
    $X\in\{M_3,\dots,B_g\}$.  From equations \eqref{eq:generic-constraints} for
    the constraints and the definition of the matrices $A,B,C$ it follows that
    the right-hand side of this equation can be expressed as a function of
    $\psi$ and the fixed parameters that characterise the conjugacy classes
    $\grporbit[1], \grporbit[2]$. This implies that there is a map $\bq_\psi:
    \RR^2 \to \RR^3$ with $\partial_\alpha\bq_\psi=0$ such that
    \begin{equation}\label{eq:psijX}
      \{\psi, \bj_X\}_D = -\idadi{X} \bq_\psi.
    \end{equation}
    Similarly, we obtain for $i\in\{1,2\}$ and functions
    $f\in\cif(\restlogr)$:
    \begin{align*}
      &\{j_{M_i}^a, f\}_D= \smashoperator{\sum_{Y\in\{M_3,\dots,B_g\}}} (J^{R,c}_Yf+J^{L,c}_Yf) \Bigl\{
        -\tidadi{M_i}{^a_c}
        +\smashoperator{\sum_{k,j=1}^3} \Bigl[\{j_{M_i}^a, C_{k+9}\} (A^\inv)_{kj} \delta^{j-1}_c \\[-1em]
       &+\{j_{M_i}^a, C_{k+6}\} (B^\inv)_{kj} \sum_{l=1}^2 \ThetaFam{j}{M_l}{d} \tidadi{M_l}{^d_c}
          - \{j_{M_i}^a, C_{k+6}\} (B^\inv C A^\inv)_{kj} \delta^{j-1}_c
        \Bigr]
      \Bigr\}\oncsurface.
    \end{align*}
    The term inside the curly brackets on the right-hand side again depends on
    $\psi$ only, which shows that there is a map $\bq_\alpha: \RR^2 \to \RR^3$
    with $\partial_\alpha\bq_\alpha=0$ such that
    \begin{equation}\label{eq:alphapX}
      \{\alpha, f\}_D = \smashoperator{\sum_{Y\in\{M_3,\dots,B_g\}}} q_\alpha^a(J^Y_{R,a}+J^Y_{L,a}) f\oncsurface
    \end{equation}
    The remaining brackets which involve $\psi,\alpha$ and the variables
    $j^a_X$, $X\in\{M_3,\dots,B_g\}$, are obtained from the Dirac brackets of
    $j^a_{M_i}$, $i=1,2$, with $j^a_X$:
    \begin{align*}
      & \{j^a_{M_i}, j^b_X\}_D =
        - \tensor{\ee}{^b^c_d} j_X^d \Bigl[
                          -\tidadi{M_i}{^a_c}
                          + \sum_{k=7}^{12}\sum_{j=1}^3 \{j^a_{M_i}, C_k\} (D^\inv)_{kj} \delta^{j-1}_c \\[-1em]
           &\hspace{10em} + \sum_{k=7}^{12}\sum_{j=4}^6 \{j^a_{M_i}, C_k\} (D^\inv)_{kj} \sum_{l=1}^2 \ThetaFam{j-4}{M_l}{e} \tidadi{M_l}{^e_c}
        \Bigr]\oncsurface\\[-.5em]
      & - \tidadi{X}{^b_c} \Bigl[
            \sum_{k=1}^6\sum_{j=7}^{9} \{j^a_{M_i}, C_k\} (D^\inv)_{kj} \tensor{W}{^c_{j-7}}
             - \sum_{k=7}^{12}\sum_{j=7}^{9} \{j^a_{M_i}, C_k\} (D^\inv)_{kj} \tensor{W}{^c_{j-7}}
          \Bigr]\oncsurface.
    \end{align*}
    The term in the second set of square brackets depends on $\psi$ and
    $\alpha$ while the term in the first set of square brackets coincides with
    the term in the curly brackets in the expression for $\{j^a_{M_i},f\}_D$.
    This implies that there is a map $\bq_\theta: \RR^2 \to \RR^3$,
    $\partial_\alpha^2\bq_\theta=0$ such that for all $X\in\{M_3,\dots,B_g\}$:
    \begin{equation*}
      \{\alpha, \bj_X\}_D = -\idadi{X}\bq_\theta - \bq_\alpha \wedge \bj_X.
    \end{equation*}

    It remains to show that $\{\alpha, \psi\}_D = 0$ and that $\bq_\psi\wedge
    \bq_\alpha = 0$.  With the definitions
    \begin{equation*}
      \begin{aligned}
        (u_{12}^\inv, \bj_{12}) &\defeq M_1^\inv \cdot M_2^\inv, \\
        (u_R^\inv, \bj_R) &\defeq M_3^\inv\cdots M_n^\inv[A_1^\inv,B_1]\cdots [A_g^\inv,B_g],
      \end{aligned}
    \end{equation*}
    the constraints $C_7,C_8,C_9$ imply $\Tr(u_{12})=\Tr(u_R)$.  From the Dirac
    brackets \eqref{eq:alphapX} of $\alpha$ with functions
    $f\in\cif(\restlogr)$ of the holonomies $M_3,\dots,B_g$, it follows that
    the Dirac bracket of $\alpha$ and $\psi$ vanishes.  Moreover, the
    constraint functions $C_1,C_2,C_3$ imply $\bj_{12} = -\Ad(u_{12}^\inv)
    \bj_R$ on $\csurface$ and it follows from \eqref{eq:psijX} that
    \begin{equation}\label{eq:psi-j12}
      0 = \{\psi, \bj_{12}\}_D = -\Ad(u_{12}^\inv) \{\psi, \bj_R\}_D = -\idadi{12} \bq_\psi.
    \end{equation}
    For each function $g\in\cif(\logr)$, we have two associated functions
    $g_{\RR^2}\in \cif(\RR^2)$, $g_{\RR^2}(\psi,\alpha)\defeq g(u_{12}^\inv)$
    and $\bar g\in\cif(\restlogr)$, $\bar g(u_{M_3},\dots,u_{B_g})\defeq
    g(u_R)$. With the identity $\{\psi,\alpha\}_D= 0$, we obtain
    \begin{equation*}
      0 = \{\alpha,  g_{\RR^2}\}_D= \{\alpha, \bar g\}_D= \smashoperator{\sum_{Y\in\{M_3,\dots,B_g\}}} q^a_\alpha (J_{R,a}^Y+J_{L,a}^Y)\bar g.
    \end{equation*}
    Together with \eqref{eq:psi-j12}, this implies that both,
    $\exp(q_\alpha^aJ_a)$ and $\exp(q_\psi^a J_a)$, stabilise $u_{12}$ and
    hence $\bq_\psi \wedge \bq_\alpha = 0$.

  \item To prove the second part of the theorem, we explicitly compute the
    Dirac brackets of the variables $\bj_X$ for $X \in \{M_3, \dots, B_g\}$ and
    functions $f\in\cif(\restlogr)$ from expressions
    \eqref{eq:generic-dirac-general}.  To determine the brackets $\{\bj_X,
    f\}_D$, we note that $\{j^a_X, C_{i+9}\} = 0$ for all $i\in\{1,2,3\}$,
    which implies
    \begin{align}\label{eq:j-f-bracket}
      &\{j_X^a, f\}_D = \{j_X^a, f\}\oncsurface - \tidadi{X}{^a^e} \, V_{ed} \smashoperator{\sum_{Y\in\{M_3,\dots,B_g\}}} (J^{R,d}_Y+J^{L,d}_Y) f\oncsurface
        \quad\text{with} \\[-.3em]
      &V_{ed} \defeq
            \tensor{W}{_e^f}(B^\inv C A^\inv)_{f+1,d+1}\oncsurface
            - \tensor{W}{_e^f} \sum_{j=1}^3 (B^\inv)_{f+1,j} \smashoperator{\sum_{i =1}^2}
                \ThetaFam{j}{M_i}{a} \tidad{M_i}{_d^a}\oncsurface.\nonumber
    \end{align}
    As none of the terms in the expression for $V$ depend on $\alpha$, it gives
    rise to a map $V:\RR^2\to \Mat(3, \RR)$ that satisfies $\partial_\alpha
    V=0$.  Similarly, we obtain
    \begin{align}\label{eq:j-j-bracket}
      \{j^a_X, j^b_Y\}_D &= \{j^a_X,j^b_Y\}\oncsurface
         + \idadi{X}^{ad} \, V_{dg} \, \tensor{\ee}{^g^b_f} j_Y^f \oncsurface\nonumber
         - \idadi{Y}^{bd} \, V_{dg} \, \tensor{\ee}{^g^a_f} j_X^f\oncsurface \\ &\hspace{5em}
         + \idadi{X}^{ac} \idadi{Y}^{bd} U_{cd}\oncsurface
    \end{align}
    with $U_{cd} \defeq \tensor{W}{_c^e} \tensor{W}{_d^f} (D^\inv)_{e+7,f+7}$
    for all $c,d\in\{0,1,2\}$. The matrix $U$ depends only on the parameters
    $\psi$ and $\alpha$, and its dependence on $\alpha$ is at most
    linear. Moreover, it follows directly from the definition of the matrix $D$
    that $U$ is antisymmetric. This allows us to expand $U$ in a basis:
    $U^{ab}=\ee^{abc}m_c$ with $\bm:\RR^2\to \RR^3$, $\partial_\alpha^2\bm=0$.

  \item By inserting the expressions \eqref{eq:vector-fields},
    \eqref{eq:action-of-vector-fields} for the left- and right-invariant vector
    fields on $\pogr$ into the Poisson bivector \eqref{eq:fr-bivector} together
    with the expression for $r(\psi,\alpha)$, one obtains after some
    computations expressions \eqref{eq:j-f-bracket},
    \eqref{eq:j-j-bracket}. This proves the claim.
  \end{enumerate}
\end{proof}

Theorem \ref{thm:generic-dirac-bracket} gives explicit expressions for the
Dirac bracket for a rather general set of gauge fixing conditions. This
generalises the results from \cite{Meusburger:2011aa}, which investigates
specific gauge fixing conditions of this type.  Given the fact that the Dirac
bracket is obtained from six first-class constraints with six associated gauge
fixing conditions and hence involves inverting a $(12\times 12)$-Dirac matrix,
its structure is surprisingly simple. This is partly due to the restriction
that the gauge fixing conditions are adapted to the tangent bundle structure of
$\pogr=T\logr$.

\subsection{The Dirac bracket and the classical dynamical Yang-Baxter equation}

The Dirac bracket in Theorem \ref{thm:generic-dirac-bracket} defines a Poisson
structure on the constraint surface $\csurface= C^\inv(0)$ which can be
identified with a subset of $\RR^2\times\restpogr$. However, this
identification is implicit, and it is cumbersome to give an explicit
parametrisation of this subset for general gauge fixing conditions. For this
reason, we consider in the following the bracket on $\RR^2\times\restpogr$
defined by the expressions in Theorem \ref{thm:generic-dirac-bracket}.

\begin{definition}\label{def:extended-dirac-bracket}
We denote by $\{\,,\,\}_D$ the antisymmetric bilinear function $\{\,,\,\}_D:
C^\infty(\RR^2\times \restpogr)\times C^\infty(\RR^2\times \restpogr)\to
C^\infty (\RR^2\times \restpogr)$ that takes the form described in Theorem
\ref{thm:generic-dirac-bracket}. With $\RR^2$ parametrised by $\psi,\alpha$ and
the different copies of $\pogr$ labelled by
$\{M_3,\dots,M_n,A_1,B_1,\dots,A_g,B_g\}$, this bracket is given by:
\begin{enumerate}
\item $\{\psi,\alpha\}_D=0$, and for all $f\in\cif(\restlogr)$,
  $X,Y\in\{M_3,\dots,B_g\}$:
  \begin{equation}\label{eq:extended-dirac-bracket-for-gauge-fixed}
    \begin{aligned}
      \{\psi, f \}_D &= 0, &
      \{\psi, \bj_X \}_D &= -\idadi{X} \, \bq_\psi, \\
      \{\alpha, f\}_D &= \smashoperator{\sum_{Y\in\{M_3,\dots,B_g\}}} q_\alpha^a(J_{R,a}^Y+J_{L,a}^Y)f, &
      \{\alpha, \bj_X\}_D &= -\idadi{X} \bq_\theta - \bq_\alpha \wedge \bj_X,
    \end{aligned}
  \end{equation}
  with $\bq_\psi, \bq_\alpha,\bq_\theta: \RR^2 \to \RR^3$ satisfying $\bq_\psi
  \wedge \bq_\alpha = 0$ and
  $\partial_\alpha\bq_\psi=\partial_\alpha\bq_\alpha=\partial_\alpha^2
  \bq_\theta= 0$.

\item For all functions $F, G \in \cif(\restpogr)$: $\{F, G\}_D =
  \frrestbivector(\diffd F \otimes \diffd G)$, where $\frrestbivector$ is the
  Poisson bivector \eqref{eq:fr-bivector} and $r: \RR^2 \to \poal \oo \poal$ is
  given by
  \begin{equation}\label{eq:extended-dirac-r}
    r (\psi,\alpha)= P_a \otimes J^a - V^{bc}(\psi)(P_b \otimes J_c - J_c \otimes P_b) + \ee^{bcd}m_d(\psi,\alpha) P_b \otimes P_c,
  \end{equation}
  with a map $V:\RR \to \Mat(3, \RR)$ that does not depend on $\alpha$ and a
  vector-valued function $\bm:\RR^2\to \RR^3$ satisfying
  $\partial_\alpha^2\bm=0$.
\end{enumerate}
\end{definition}

This bracket has a particularly simple structure. The two variables $\psi$ and
$\alpha$ Poisson-commute, and their Dirac brackets with functions on
$\restpogr$ are given by three functions
$\bq_\psi,\bq_\alpha,\bq_\theta:\RR^2\to\RR^3$. The Dirac bracket of two
functions on $\restpogr$ is again given by the Poisson bivector
\eqref{eq:fr-bivector}. The only difference is that the classical $r$-matrix
$r=P_a\oo J^a$ in the Poisson bivector $\frrestbivector$ is now replaced by the
map $r:\RR^2 \to \poal \oo \poal$ that depends on the variables $\psi$ and
$\alpha$.

Note that it is a priori not guaranteed that the bracket $\{\,,\,\}_D$ on
$\RR^2\times\restpogr$ satisfies the Jacobi identity. The Dirac gauge fixing
formalism only guarantees that this is the case on the constraint surface
$\csurface=C^\inv(0)\subset \RR^2\times\restpogr$. Moreover, it is natural to
ask how the Jacobi identity is encoded in the structures that characterise the
bracket in Definition \ref{def:extended-dirac-bracket}: the map $r:\RR^2\to
\poal \oo \poal$ and the vector-valued functions $\bq_\psi,
\bq_\alpha,\bq_\theta: \RR^2 \to \RR^3$.  As the classical Yang-Baxter equation
for the $r$-matrix in the Poisson bivector \eqref{eq:fr-bivector} ensures that
the associated bracket satisfies the Jacobi identity, it is natural to expect
that the Jacobi identity for the Dirac bracket follows from an analogous
property of the map $r$. This suggests that $r$ should be related to solutions
of the classical {\em dynamical} Yang-Baxter equation and hence to classical
dynamical $r$-matrices for the Lie algebra $\poal$.

This intuition is also supported by the fact that Fock and Rosly's Poisson
structure is related to certain Poisson structures from the theory of
Poisson-Lie groups \cite{Fock:1998aa, Alekseev:1995ab}.  It is shown in
\cite{Feher:2001aa, Feher:2004aa} that Dirac gauge fixing in the context of
Poisson-Lie groupoids is linked to classical dynamical $r$-matrices.  Note,
however, that our case is more involved. While the references
\cite{Feher:2001aa, Feher:2004aa} consider a gauge fixing procedure for a
generalisation of the Sklyanin bracket in the context of Poisson-Lie groupoids,
our Poisson structure involves several copies of the dual Poisson-Lie structure
and the Heisenberg double Poisson structure which interact in a non-trivial
way. Moreover, the gauge fixing conditions we consider are associated with two
punctures and hence with two non-Poisson-commuting dual Poisson-Lie structures
whose Poisson brackets with the remaining punctures and handles do not
vanish. Nevertheless, it is natural to expect that our gauge fixing procedure
should be related to solutions of the classical dynamical Yang-Baxter equation.

The concept of a classical dynamical $r$-matrix generalises the notion of
classical $r$-matrices $r\in\ag\oo\ag$ for a Lie algebra $\ag$ to maps $r: U\to
\ag\oo\ag$ that depend non-trivially on variables in $U$. The domain $U$ is an
open subset of the dual $\ah^*$ of an abelian Lie subalgebra $\ah\subset\ag$,
and the map $r$ is required to be invariant under the action of $\ah$.  Instead
of the classical Yang-Baxter equation (CYBE), the classical dynamical
$r$-matrix is required to satisfy the classical dynamical Yang-Baxter equation
(CDYBE).  The latter is obtained by replacing the right-hand side of the CYBE
by a term that contains the derivatives of $r$ with respect to the coordinates
on $U$.

\begin{definition}[\cite{Etingof:1998aa}]\label{def:dyn-r-matrix}
  Let $\ag$ be a finite-dimensional Lie algebra, $\ah \subset \ag$ an abelian
  Lie subalgebra, and $U \subset \ah^*$ an open subset.  A \definee{classical
    dynamical $r$-matrix} for $(\ag, \ah, U)$ is an $\ah$-invariant,
  meromorphic function $r: U \to \ag \otimes \ag$ that satisfies the
  \definee{classical dynamical Yang-Baxter equation} (CDYBE):
  \begin{equation}\label{eq:dcybe}
    [[r,r]] \defeq [r_{12}, r_{13}] + [r_{12}, r_{23}] + [r_{13}, r_{23}] =
      \sum_{i=1}^{\text{dim}\,\ah} \left(
        x_i^{(1)}\partial_{x^i}\,r_{23} - x_i^{(2)}\partial_{x^i}\,r_{13} + x_i^{(3)}\partial_{x^i}\,r_{12}
      \right),
  \end{equation}
  where $\{x_i\}_{i=1,\dots,\text{dim}\, \ah}$ is a basis of $\ah$ and
  $\{x^i\}_{i=1,\dots,\text{dim}\, \ah}$ the associated dual basis of $\ah^*$.
\end{definition}

In the following, we only require the case where $\ag = \poal$ and $\ah$ is a
two-dimensional abelian Lie subalgebra of $\ag$.  We thus identify $\ah^*$ with
$\RR^2$ and parametrise it by two variables $x^1=\psi$, $x^2=\alpha$.
Moreover, we temporarily drop the requirements that the elements $x_1,x_2$ in
the CDYBE form a fixed basis of $\ah\subset\poal$ and that $r$ is invariant
under the action of $\ah$.

Instead, we investigate solutions of the CDYBE \eqref{eq:dcybe} associated with
maps $x_1,x_2:\RR^2\to\poal$ of the form $x_1=q_\psi^aP_a$ and
$x_2=q_\alpha^aJ_a+q_\theta^aP_a$ with $\bq_\psi,\bq_\alpha:\RR\to\RR^3$,
$\bq_\theta:\RR^2\to\RR^3$ satisfying $\bq_\psi\wedge\bq_\alpha=0$. Note that
this implies that $\ah(\psi,\alpha)=\Span\{x_1(\psi,\alpha),
x_2(\psi,\alpha)\}$ is a two-dimensional abelian Lie subalgebra of $\poal$ for
all values of $\psi$ and $\alpha$.  It is a Cartan subalgebra if and only if
$\bq_\psi, \bq_\alpha$ satisfy the additional requirement $\bq_\alpha^2,
\bq_\psi^2\neq 0$.  We do not assume that $r(\psi,\alpha)$ is invariant under
the subalgebra $\ah(\psi,\alpha)$.

Although such solutions of the CDYBE \eqref{eq:dcybe} do not correspond to
classical dynamical $r$-matrices in the sense of Definition
\ref{def:dyn-r-matrix}, admitting such generalised solutions allows us to apply
the CDYBE to the maps $r$ and $\bq_\psi,\bq_\alpha,\bq_\theta$ in Definition
\ref{def:extended-dirac-bracket} and to determine under which conditions they
give rise to a solution of the CDYBE.  By comparing these conditions to the
requirement that the bracket $\{\,,\}_D$ in Definition
\ref{def:extended-dirac-bracket} satisfies the Jacobi identity, we obtain the
following theorem.

\begin{theorem}\label{thm:jacobi-for-extended-dirac}
  The bracket $\{\,,\,\}_D$ in Definition \ref{def:extended-dirac-bracket}
  satisfies the Jacobi identity and hence defines a Poisson structure on
  $\restpogr$ if and only if:
  \begin{compactenum}
  \item The map $r:\RR^2\to \poal\oo\poal$ in \eqref{eq:extended-dirac-r}
    satisfies the CDYBE with $x^1=\psi$, $x^2=\alpha$ and $x_1=q_\psi^aP_a$,
    $x_2=q_\alpha^aJ_a+q_\theta^a P_a$.
  \item The following additional conditions hold:
    \begin{equation}\label{eq:q-relations}
      \left.
      \begin{aligned}
        0 &= q_\psi^a + \tensor{\ee}{^a_b_c} q_\psi^b \partial_\psi q_\psi^c + q_\psi^b \tensor{V}{_b^a} - q_\psi^a \tensor{V}{^b_b}, \\
        0 &= \tensor{\ee}{^a_d_h} q_\alpha^d V^{bh} + \tensor{\ee}{^b_d_h} q_\alpha^d V^{ah} + \tensor{\ee}{_c_d_e} q_\alpha^c V^{de} \eta^{ab} - \tensor{\ee}{^b_d_e} q_\alpha^a V^{de} + q_\alpha^a \partial_\alpha q_\theta^b - q_\psi^b \partial_\psi q_\alpha^a, \\
        0 &= q_\theta^a + \tensor{\ee}{^a_b_c} q_\theta^b \partial_\alpha q_\theta^c + \tensor{\ee}{^a_b_c} q_\psi^b \partial_\psi q_\theta^c - \tensor{\ee}{^a_b_c} m^b q_\alpha^c + q_\theta^d \tensor{V}{_d^a} - q_\theta^a \tensor{V}{_d^d}.
      \end{aligned}
      \quad\right\}
    \end{equation}
  \end{compactenum}
\end{theorem}

\begin{proof} $\quad$
\begin{enumerate}
\item As a first step, we show that a map $r: \RR^2 \to \poal \otimes \poal$ of
  the form \eqref{eq:extended-dirac-r} is a solution of the CDYBE with
  $x^1=\psi$, $x^2=\alpha$, $x_1=q_\psi^a P_a:\RR\to\RR^3$, $x_2=q_\alpha^a
  J_a+q_\theta^a P_a:\RR^2\to\RR^3$ if and only if it satisfies the equations
  \begin{equation}\label{eq:dcybe-with-UV}
    \left.
    \begin{aligned}
      0 &= \Upsilon^{abc} \defeq q_\alpha^a \ee ^{bcd} \partial_\alpha m_d - q_\psi^b \partial_\psi V^{ca} + q_\psi^c \partial_\psi V^{ba} \\
        &\qquad\qquad\quad - V^{bd}V^{cg}\tensor{\ee}{_d_g^a} - V^{da}V^{cg}\tensor{\ee}{_d_g^b} + V^{da}V^{bg}\tensor{\ee}{_d_g^c} - V^{da} \tensor{\ee}{^b^c_d}, \\
     0 &= \Omega\defeq
       \bq_\psi \cdot \partial_\psi \bm + \bq_\theta \cdot \partial_\alpha \bm +\bw \cdot \bm
       \qquad\text{with}\quad \ee^{abc}w_c=V^{ab}-V^{ba}.
    \end{aligned}
    \qquad\right\}
  \end{equation}
  Inserting expression \eqref{eq:extended-dirac-r} for $r$ into the left-hand
  side of the CDYBE \eqref{eq:dcybe} and using expressions
  \eqref{eq:poincare-algebra-bracket} for the Lie bracket of $\poal$, we obtain
  \begin{equation*}
    \begin{split}
      &[[r, r]] = -\bw \cdot \bm \, \ee^{abc} P_a \otimes P_b \otimes P_c + {} \\
        &\bigl[V^{bd}V^{cg}\tensor{\ee}{_d_g^a} \!+\! V^{da}V^{cg}\tensor{\ee}{_d_g^b}\! -\! V^{da}V^{bg}\tensor{\ee}{_d_g^c} \!+ \!V^{da} \tensor{\ee}{^b^c_d}\bigr]
        (J_a\!\!\otimes\!\!P_b\!\!\otimes\!\!P_c \!-\! P_b\!\!\otimes\!\!J_a\!\!\otimes\!\!P_c \!+\! P_b\!\!\otimes\!\!P_c\!\!\otimes\!\!J_a).
    \end{split}
  \end{equation*}
  Setting $x^1 = \psi$, $x^2 = \alpha$, $x_1 = q_\psi^a P_a$, $x_2 = q_\alpha^a
  J_a + q_\theta^a P_a$ and using $\partial_\alpha V = 0$, we find that the
  right-hand side of the CDYBE is given by:
 \begin{equation*}
    \begin{split}
      &\sum_{i=1}^2 x_i^{(1)} \partial_{x_i} r_{23} -
                    x_i^{(2)} \partial_{x_i} r_{13} +
                    x_i^{(3)} \partial_{x_i} r_{12} =
        (\bq_\psi\cdot\partial_\psi\bm + \bq_\theta\cdot\partial_\alpha\bm)\ee^{abc} P_a \!\otimes\! P_b \!\otimes\! P_c + {} \\
      & \left(q_\alpha^a \ee^{bcd}\partial_\alpha m_d-
              q_\psi^b \partial_\psi V^{ca} +
              q_\psi^c \partial_\psi V^{ba}\right)
        (J_a \otimes P_b \otimes P_c - P_b \otimes J_a \otimes P_c + P_b \otimes P_c \otimes J_a).
    \end{split}
  \end{equation*}
  A comparison of the coefficients in these two expressions then yields
  equations \eqref{eq:dcybe-with-UV}.

\item To determine under which conditions the bracket in Definition
  \ref{def:extended-dirac-bracket} satisfies the Jacobi identity, we consider
  the variables $\psi, \alpha$ , functions $h\in\cif(\restlogr)$ and the
  variables $ j^a_X$ for $X \in \{M_3, \dots, B_g\}$.  The structure of the
  Poisson algebra in Theorem \ref{thm:generic-dirac-bracket} allows us to
  reduce the proof to six cases which are distinguished by the number of
  variables $j_X^a$, $\psi$, $\alpha$ in the brackets.
  \begin{enumerate}
  \item For cyclic sums over brackets of the form $\{h, \{j^b_Y,j^c_Z\}_D\}_D$
    with $Y,Z\in\{M_3,\dots,B_g\}$, we obtain
    \begin{multline*}
      \{h, \{j_Y^b, j_Z^c\}_D\}_D + \{j_Y^b, \{j_Z^c, h\}_D\}_D + \{j_Z^c, \{h, j_Y^b\}_D\}_D \\
      = \tidadi{Y}{^b_d} \tidadi{Z}{^c_e}  \Upsilon^{deg} \smashoperator{\sum_{Y\in\{M_3,\dots,B_g\}}} (R^g_Y+L^g_Y)h,
    \end{multline*}
    where $\Upsilon^{deg}$ is the term in the first equation of
    \eqref{eq:dcybe-with-UV}.  Consequently, it vanishes if $r$ satisfies the
    CDYBE.

  \item For cyclic sums over brackets of the form $\{j_X^a,
    \{j^b_Y,j^c_Z\}_D\}_D$ with $X,Y,Z\in\{M_3,\dots,\allowlinebreakhere
    B_g\}$, we have
    \begin{equation*}
      \begin{split}
        &\{j_X^a, \{j_Y^b, j_Z^c\}_D\}_D + \{j_Y^b, \{j_Z^c, j_X^a\}_D\}_D + \{j_Z^c, \{j_X^a, j_Y^b\}_D\}_D \\
        &\qquad = \tidadi{Y}{^b_f} \tidadi{Z}{^c_e} \, \tensor{\ee}{^a_d_g} \, j_X^d \Upsilon^{efg} \\
        &\qquad + \tidadi{X}{^a_e} \tidadi{Z}{^c_f} \, \tensor{\ee}{^b_d_g} \, j_Y^d \Upsilon^{efg} \\
        &\qquad + \tidadi{X}{^a_f} \tidadi{Y}{^b_e} \, \tensor{\ee}{^c_d_g} \, j_Z^d \Upsilon^{efg} \\
        &\qquad + \tidadi{X}{^a_d} \tidadi{Y}{^b_e} \tidadi{Z}{^c_f} \ee^{def} \Omega,
      \end{split}
    \end{equation*}
    where $\Upsilon^{efg}$ and $\Omega$ are, respectively, the terms in the
    first and second lines of \eqref{eq:dcybe-with-UV}.  This shows that the
    Jacobi identity for brackets of this type is satisfied if and only if $r$
    is a solution of the CDYBE.

  \item The remaining cases involve cyclic sums over brackets of the form
    $\{\psi, \{j^a_X. j^b_Y\}_D\}_D$, $\{\psi,\{\alpha,j^a_X\}_D\}_D$,
    $\{\alpha,\{h,j^b_Y\}_D\}_D$ and $\{\alpha,\{j^a_X,j^b_Y\}_D\}_D$ with
    $X,Y\in\{M_3,\dots,B_g\}$.  A direct calculation along the same lines as in
    cases (a) and (b) shows that the Jacobi identity is satisfied for brackets
    of this type if and only if the identities in \eqref{eq:q-relations} hold.
    \end{enumerate}
  \end{enumerate}
\end{proof}

Theorem \ref{thm:jacobi-for-extended-dirac} gives a direct link between Poisson
structures on $\RR^2\times\restpogr$ of the form in Definition
\ref{def:extended-dirac-bracket} and solutions of the CDYBE. As is apparent in
the proof, the CDYBE is a necessary and sufficient condition which ensures that
the Poisson brackets of functions $F, G\in\cif(\restpogr)$ satisfy the Jacobi
identity for all values of $\psi$ and $\alpha$. The additional conditions
\eqref{eq:q-relations} ensure that the Jacobi identity also holds for mixed
brackets involving the variables $\psi,\alpha$ as well as functions
$F\in\cif(\restpogr)$.  We will show in the next section that these conditions
have a direct geometrical interpretation. They allow one to locally transform a
solution $r:\RR^2\to\poal\oo \poal$ of the CDYBE into a classical dynamical
$r$-matrix in the sense of Definition \ref{def:dyn-r-matrix} that is invariant
under a fixed Cartan subalgebra $\ah\subset\poal$.

\subsection{Examples of solutions}

The conditions \eqref{eq:dcybe-with-UV} that characterise the classical
dynamical Yang-Baxter equation and the supplementary conditions
\eqref{eq:q-relations} in Theorem \ref{thm:jacobi-for-extended-dirac} are quite
complicated.  It is therefore not obvious to determine solutions of these
equations.  In the following, we show that the specific gauge fixing conditions
investigated in \cite{Meusburger:2011aa} give rise to a solution of the CDYBE
that also satisfies the additional conditions \eqref{eq:q-relations} in Theorem
\ref{thm:jacobi-for-extended-dirac}. We also determine a simplified standard
set of solutions of these equations that are classical dynamical $r$-matrices
in the sense of Definition \ref{def:dyn-r-matrix}.

The publication \cite{Meusburger:2011aa} investigates a specific set of gauge
fixing conditions of the type discussed in Section
\ref{subsec:constraints-and-gauge-fixing}, which is motivated by its direct
physical interpretation in the application to the Chern-Simons formulation of
(2+1)-gravity.  These gauge fixing conditions consider the case where the two
gauge-fixed holonomies $M_1,M_2$ are restricted to conjugacy classes
\begin{equation*}
  \grporbit[j]=\{h\cdot \exp(-\mu_j J_0-s_j P_0)\cdot h^\inv \mid h\in \pogr\} \qquad
  \forall j=1,2,
\end{equation*}
with $\mu_1,\mu_2\in(0,2\pi)$, $s_1,s_2\in\RR$.  The resulting Dirac bracket is
determined in \cite{Meusburger:2011aa}. It takes the form of Theorem
\ref{thm:generic-dirac-bracket} and Definition \ref{def:extended-dirac-bracket}
with $r:\RR^2\to\poal\oo\poal$ given by
\begin{equation}\label{eq:special-w-m}
  \left.
  \begin{aligned}
    &r(\psi,\alpha)= \tfrac 1 2 \left(P_a \!\otimes\! J^a\!+\!J^a\!\otimes\! P_a\right) \!-\! \tfrac 1 2 \ee^{abc} w_c(\psi)(P_a \!\otimes\! J_b \!-\! J_b \!\otimes\! P_a) \!+\! \ee^{abc}m_c(\psi,\alpha) P_a \!\otimes\! P_b,\\
    &\bw(\psi)=\cot\tfrac{\mu_1} 2\, \be_0+\cot\tfrac{\mu_1} 2 \coth\psi \, \be_1-\coth\psi \, \be_2,\\
    &\bm(\psi,\alpha)=s_1/(4\sin^2\tfrac{\mu_1} 2)\,\be_0 + s_1\coth\psi/(4\sin^2\tfrac{\mu_1} 2)\,\be_1 + \tfrac12 \alpha \, \partial_\psi \bw(\psi).
  \end{aligned}
  \;\;\right\}
\end{equation}
The associated maps $\bq_\psi,\bq_\alpha,\bq_\theta:\RR^2\to\RR^3$ take the
form
\begin{equation}\label{eq:special-q-psi-alpha}
\left.
\begin{aligned}
  &\bq_\psi(\psi,\alpha)=-\bq_\alpha(\psi,\alpha)=\tfrac 1 2 ({\coth\psi \cot\tfrac{\mu_1}{2}+ \cot\tfrac{\mu_2}{2}/\sinh\psi})\,\be_0+\tfrac 1 2 \cot\tfrac{\mu_1} 2\,\be_1-\tfrac 1 2 \, \be_2,\\
  &\bq_\theta (\psi,\alpha)\!=\!\left[
      \frac{s_1\coth\psi}{4\sin^2\tfrac{\mu_1}{2}} \!+\!
      \frac{s_2}{4\sin^2\tfrac{\mu_2}{2}\sinh\psi} \!-\!
      \frac{\alpha (\cot\frac{\mu_1}{2} \!+\! \cosh\psi\cot\frac{\mu_2}{2})}{2\sinh^2\psi}\!
    \right]\be_0 +
    \frac{s_1}{4\sin^2\tfrac{\mu_1}{2}}\,\be_1.
\end{aligned}
\;\right\}
\end{equation}

We will now show that this defines a solution of the CDYBE which also satisfies
the additional conditions \eqref{eq:q-relations}.

\begin{lemma}\label{lem:dcybe-for-special-r}
  The map $r: \RR^2 \to \poal\oo\poal$ in \eqref{eq:special-w-m} is a solution
  of the CDYBE with $x^1=\psi$, $x^2=\alpha$, $x_1=q_\psi^a P_a$,
  $x_2=q_\alpha^a J_a+q_\theta^a P_a$, $\bq_\psi,\bq_\alpha,\bq_\theta$ as in
  \eqref{eq:special-q-psi-alpha}, and satisfies the additional conditions
  \eqref{eq:q-relations}.  The maps $x_1,x_2:\RR^2\to\poal$ define a Cartan
  subalgebra $\ah(\psi,\alpha)$ for all values of $\psi$ for which
  $\bq_\psi^2(\psi), \bq_\alpha^2(\psi) \neq 0$.
\end{lemma}

\begin{proof} $\quad$
\begin{enumerate}
\item That the maps $x_1,x_2:\RR^2\to\poal$ define a two-dimensional abelian
  Lie subalgebra of $\poal$ follows directly from the condition
  $\bq_\psi\wedge\bq_\alpha=0$. One finds that this Lie subalgebra is a Cartan
  subalgebra for all values of $\psi$ for which $\bq_\psi^2(\psi) \neq 0$.

\item That $r$ solves the CDYBE can be shown with a direct calculation.
  Inserting the expressions \eqref{eq:special-w-m} for $\bw,\bm$ and
  expressions \eqref{eq:special-q-psi-alpha} for $\bq_\psi,\bq_\theta$ into the
  left-hand-side of the second equation in \eqref{eq:dcybe-with-UV}, one finds
  after some computations that this expression vanishes.  To verify that the
  first set of equations in \eqref{eq:dcybe-with-UV} is satisfied, we note that
  for maps $V:\RR\rightarrow\Mat(3, \RR)$ of the form $V^{ab}(\psi) = \frac12
  \eta^{ab} + \frac12 \tensor{\ee}{^a^b_c} w^c(\psi)$ with $\bw: \RR \to
  \RR^3$, the first line of \eqref{eq:dcybe-with-UV} is equivalent to the
  following conditions
  \begin{equation}\label{eq:dcybe-with-simple-UV}
    1+\bw^2+2\bq_\psi\cdot \partial_\psi \bw=0,\qquad
    \partial_\psi\bw\wedge \partial_\alpha\bm=0,\qquad
    \partial_\alpha m^a q_\alpha^b=-\tfrac 1 2 \partial_\psi w^a q_\psi^b.
  \end{equation}
  Setting $\bq_\alpha=-\bq_\psi$ and inserting expressions
  \eqref{eq:special-w-m} for $\bw,\bm$ together with expression
  \eqref{eq:special-q-psi-alpha} for $\bq_\psi$ into
  \eqref{eq:dcybe-with-simple-UV}, one verifies the first condition in
  \eqref{eq:dcybe-with-UV}.

\item To show that $r$ and $\bq_\psi,\bq_\alpha,\bq_\theta$ satisfy the
  additional conditions \eqref{eq:q-relations}, we note that for matrices $V$
  of the form $V^{ab}(\psi)=\tfrac 1 2 \eta^{ab}+\tfrac 1 2 \ee^{abc}
  w^c(\psi)$, these conditions reduce to the following set of equations
  \begin{equation}
  \left.
  \begin{aligned}\label{eq:q-relations-simple}
    0&=\bq_\psi\wedge  (\partial_\psi\bq_\psi-\tfrac 1 2 \bw),\\
    0&=\tfrac 1 2 (q_\alpha^bw^a-q_\alpha^aw^b)+q^a_\alpha\partial_\alpha q^b_\theta-q^b_\psi\partial_\psi q^a_\alpha\quad\forall a,b\in\{0,1,2\},\\
    0&=\bq_\theta\wedge(\partial_\alpha\bq_\theta-\tfrac 1 2 \bw)+\bq_\psi\wedge \partial_\psi \bq_\theta+\bq_\alpha\wedge\bm.
  \end{aligned}
  \qquad\right\}
  \end{equation}
  Inserting expressions \eqref{eq:special-w-m} for $\bw,\bm$ and expressions
  \eqref{eq:special-q-psi-alpha} for $\bq_\psi,\bq_\alpha,\bq_\theta$ into
  these equations, one finds that they are indeed satisfied.
\end{enumerate}
\end{proof}

Note that the resulting solution of the CDYBE in \cite{Meusburger:2011aa} is
{\em not} a classical dynamical $r$-matrix in the sense of Definition
\ref{def:dyn-r-matrix}. While Definition \ref{def:dyn-r-matrix} requires the
choice of an abelian Lie subalgebra $\ah\subset \poal$ and an identification of
the two variables in the solution with its dual, the abelian Lie subalgebra
$\ah(\psi,\alpha)=\Span\{ q_\psi^aP_a, q_\alpha^aJ_a+q_\theta^aP_a\}$
associated with the above solution varies with $\psi$ and $\alpha$.  A direct
calculation shows that depending on the value of $\psi$, the Lie subalgebra
$\ah(\psi,\alpha)$ is conjugate either to the Cartan subalgebra
$\ah_a=\Span\{J_0,P_0\}$ for $\bq_\psi^2(\psi)>0$, to the Cartan subalgebra
$\ah_b=\Span\{J_1,P_1\}$ for $\bq_\psi^2(\psi)<0$ or to the two-dimensional Lie
subalgebra $\ah_c=\Span\{J_0+J_1,P_0+P_1\}$ for $\bq_\psi^2(\psi)=0$.

The solution therefore combines solutions of the CDYBE that are associated with
different, non-conjugate two-dimensional Lie subalgebras of $\poal$. To show
that the existence of solutions associated with different Lie subalgebras is a
generic phenomenon and not a consequence of the specific gauge fixing
conditions in \cite{Meusburger:2011aa}, we determine a simple set of solutions
of a similar form.

\begin{lemma}\label{lem:simple-gen-solutions}
For all $c\in\RR$, $\gamma\in\cif(\RR)$, the map $r:\RR^2\to\poal\oo
\poal$ given by
\begin{equation*}
  r(\psi,\alpha)=\tfrac 1 2 (P_a\!\oo\!J^a + J^a\!\oo\!P_a) - \tensor{\ee}{^a^b_c}\,\partial_\psi q_\psi^c(\psi)(P_a{\oo} J_b - J_b {\oo} P_a) - \alpha\, \tensor{\ee}{^a^b_c}\,\partial_\psi^2 q_\psi^c (\psi)\, P_a {\oo} P_b
\end{equation*}
is a solution of the CDYBE with $x_1=\psi$, $x_2=\alpha$, $x^1=q_\psi^aP_a$,
$x^2=q_\alpha^aJ_a+q_\theta^aP_a$ and
\begin{equation*}
  \bq_\psi(\psi)=\bq_\alpha(\psi)=\gamma(\psi) \be_0+\sqrt{\gamma^2(\psi)+\tfrac 1 4 (\psi-c)^2} \,\be_1,\qquad
  \bq_\theta(\psi,\alpha)=\alpha\,\partial_\psi\bq_\psi(\psi).
\end{equation*}
\end{lemma}

\begin{proof}
This follows by a direct calculation.  As $r$ is of a form similar to the
solution in Lemma~\ref{lem:dcybe-for-special-r} with
$\bw=2 \partial_\psi\bq_\psi$, $\bm=-\alpha\, \partial_\psi^2 \bq_\psi$,
inserting these expressions into \eqref{eq:q-relations-simple} shows directly
that the conditions \eqref{eq:q-relations} are satisfied. The CDYBE then
reduces to the requirement $1+2\partial_\psi^2(\bq_\psi^2)=0$, which is
verified by a simple computation.
\end{proof}

Note, however, that the CDYBE \eqref{eq:dcybe-with-UV} and the additional
requirements \eqref{eq:q-relations} also admit solutions which are associated
with fixed Cartan subalgebras of $\poal$ and define classical dynamical
$r$-matrices in the sense of Definition \ref{def:dyn-r-matrix}.  To obtain such
solutions, we set $\bq_\theta(\psi,\alpha)=0$ and either
$\bq_\psi(\psi)=\bq_\alpha(\psi)=\be_0$ or
$\bq_\psi(\psi)=\bq_\alpha(\psi)=\be_1$ for all admissible values of
$\psi$. The conditions \eqref{eq:q-relations-simple} then reduce to the
requirements $\bw,\bm\in\Span\{\bq_\psi\}$, and the expressions
\eqref{eq:dcybe-with-simple-UV} to $\partial_\alpha\bm=-\tfrac 1
2 \partial_\psi\bw$, $1+\bw^2+2\bq_\psi\cdot\partial_\psi\bw=0$.  From this,
one then obtains two solutions associated with the Cartan subalgebras
$\ah_a=\Span\{J_0,P_0\}$ and $\ah_b=\Span\{J_1,P_1\}$ in $\poal$.

\begin{lemma}\label{lem:dcybe-solutions}
Two solutions of the CDYBE with $x^1=\psi$, $x^2=\alpha$, $x_1=q_\psi^aP_a$ and
$x_2=q_\alpha^aJ_a+q^a_\theta P_a$ that also satisfy the additional conditions
\eqref{eq:q-relations} are given by
\begin{compactenum}[\quad a)]
\item $\bq_\psi=\bq_\alpha=\be_0$, $\bq_\theta=0$ and $r:{(-\tfrac \pi 2,\tfrac \pi 2)}\times \RR\to\poal\oo\poal$,
  \begin{equation*}
    r_a(\psi,\alpha)=\tfrac 1 2 (P_a\!\oo\!J^a+J^a\!\oo\!P_a) + \tfrac12 \tan\tfrac\psi 2\left(P_1{\wedge} J_2 - P_2{\wedge} J_1\right) + \frac{\alpha}{4\cos^2\tfrac \psi 2} P_1{\wedge} P_2,
  \end{equation*}

\item $\bq_\psi=\bq_\alpha=\be_1$, $\bq_\theta=0$ and $r:
  \RR^2\to\poal\oo\poal$,
  \begin{equation*}
    r_b(\psi,\alpha)=\tfrac 1 2 (P_a\!\oo\!J^a+J^a\!\oo\!P_a) + \tfrac 1 2 \tanh\tfrac\psi 2 \left(P_2{\wedge} J_0 - P_0{\wedge} J_2\right) + \frac{\alpha}{4\cosh^2\tfrac \psi 2} P_2{\wedge} P_0,
  \end{equation*}
\end{compactenum}
where $X\wedge Y:=X\oo Y-Y\oo X$.  They are classical dynamical $r$-matrices as
in Definition \ref{def:dyn-r-matrix} for, respectively, the Cartan subalgebras
$\ah_a=\Span\{J_0,P_0\}$ and $\ah_b=\Span\{J_1,P_1\}$.
\end{lemma}

Lemma \ref{lem:dcybe-solutions} provides us with two particularly simple
classical dynamical $r$-matrices for $\poal$. We will show in the next section
that every solution of the CDYBE of the form \eqref{eq:extended-dirac-r} which
satisfies the additional conditions \eqref{eq:q-relations} can be transformed
into one of these two solutions for all values of $\psi$ for which either
$\bq_\psi^2(\psi),\bq_\alpha^2(\psi)>0$ (case a) or
$\bq_\psi^2(\psi),\bq_\alpha^2(\psi)<0$ (case b).

One might wonder if there are similar solutions of the CDYBE and conditions
\eqref{eq:q-relations} for which $\bq_\psi, \bq_\alpha$ are fixed lightlike
vectors that do not depend on $\psi$ and $\alpha$. However, it turns out that
such solutions do not exist. This appears to be linked to the fact that the
vectors $\bq_\psi, \bq_\alpha$ associated with the solutions in Lemma
\ref{lem:dcybe-for-special-r} and Lemma \ref{lem:simple-gen-solutions} are
spacelike or timelike for values of $\psi$ in certain open intervals of $\RR$,
but can become lightlike only for a very specific discrete set of values of
$\psi$. This again suggests that the variation of $\bq_\psi, \bq_\alpha$ with
$\psi$ is a generic feature of the gauge fixing procedure, and that there are
no gauge fixing conditions that allow one to obtain a Poisson structure
determined by vectors $\bq_\psi, \bq_\alpha$ that are lightlike for all $\psi$.
We have the following lemma:

\begin{lemma}
There are no simultaneous solutions of the CDYBE \eqref{eq:dcybe} and
conditions \eqref{eq:q-relations} for which $\bq_\psi,\bq_\alpha,\bq_\theta$
are constant vectors with $\bq_\psi\wedge\bq_\alpha=0$ and
$\bq_\psi^2=\bq_\alpha^2=0$.
\end{lemma}

\begin{proof}
Suppose that $r:\RR^2\to\poal\oo\poal$ is of the form in
Definition~\ref{def:extended-dirac-bracket} with $V^{ab}(\psi)=\tfrac 1 2
\eta^{ab}+Q^{ab}(\psi)+\tfrac12\ee^{abc} w_c(\psi)$, where $Q:\RR\to\Mat(3,
\RR)$ is symmetric. Then conditions~\eqref{eq:q-relations} imply
$Q^{02}=Q^{12}=Q^{22}=0$, $Q^{00}=Q^{01}=Q^{11}$ and $\bw \in
\Span\{\be_0+\be_1\}$.  Inserting this into the first equation of
\eqref{eq:dcybe-with-UV} yields a contradiction and thus proves the claim.
\end{proof}

\section{Transformations between  Dirac brackets}
\label{sec:dynamic-poincare-trafos}

\subsection{Dynamical  Poincar\'e transformations}

The results of the previous section show how constraint functions that satisfy
the conditions in Definition \ref{def:dirac-bracket} and the requirements a and
b from Section~\ref{subsec:constraints-and-gauge-fixing} give rise to a Dirac
bracket that is given by solutions of the CDYBE.  In this section, we
investigate the transformation of the Dirac bracket under a change of
constraint functions.

For this, recall from the discussion before Theorem
\ref{thm:generic-dirac-bracket} that any set of admissible constraint functions
restricts the variables $M_1,M_2$ which parametrise the first two copies of
$\pogr$ in $\allpogr$ in such a way that they are determined uniquely by the
two conjugation invariant quantities $\psi,\alpha$ which depend only on the
product $M_2\cdot M_1$.  This suggests that for any two sets of variables
$M_1,M_2$ and $M_1',M_2'$ obtained in this way, there should be a Poincar\'e
transformation $p(\psi,\alpha)\in \pogr$ such that $M_1'=p(\psi,\alpha)\cdot
M_1\cdot p(\psi,\alpha)^\inv$ and $M_2'=p(\psi,\alpha)\cdot M_2\cdot
p(\psi,\alpha)^\inv$. If the associated gauge fixing conditions satisfy
conditions a and b from Section \ref{subsec:constraints-and-gauge-fixing}, it
follows from \eqref{eq:psi-alpha-general} that one can restrict attention to
Poincar\'e transformations $p(\psi,\alpha)$ whose Lorentzian components do not
depend on $\alpha$ and whose translational components depend on $\alpha$ at
most linearly.

This is a strong motivation to investigate the transformation of the bracket in
Definition \ref{def:extended-dirac-bracket} under such Poincar\'e
transformations.  We therefore consider smooth maps
\begin{align}\label{eq:poincare-diffeo}
  \Phi^p:\; &\RR^2\times\restpogr\to \RR^2\times\restpogr,\\
  &(\psi,\alpha,M_3,\dots,B_g)\mapsto (\psi,\alpha,\; p(\psi,\alpha)\cdot M_3\cdot p(\psi,\alpha)^\inv,\dots,\;p(\psi,\alpha)\cdot B_g\cdot p(\psi,\alpha)^\inv),\nonumber
\end{align}
where $p=(g,-\Ad(g)\bt)\in\cif(\RR^2,\pogr)$ with $\partial_\alpha
g=\partial_\alpha^2\bt=0$.  We find that the transformation of the bracket
$\{\,,\}_D$ under such a dynamical Poincar\'e transformation corresponds to a
simultaneous transformation of the maps $r:\RR^2\to\poal\oo\poal$ and
$\bq_\psi,\bq_\alpha,\bq_\theta:\RR^2\to\RR^3$.

\begin{lemma}\label{lem:poinc-trafo}
  Let $\{,\}_D$ and $r: \RR^2 \to \poal \oo \poal$ be given as in Definition
  \ref{def:extended-dirac-bracket} and consider a dynamical Poincar\'e
  transformation $\Phi^p:\RR^2\times\restpogr\to \RR^2\times\restpogr$ as
  above. Then for all $F,G\in\cif(\RR^2\times\restpogr)$:
  \begin{equation*}
   \{ F\circ \Phi^p, G\circ \Phi^p\}_D=\{F,G\}^p_D\circ \Phi^p,
  \end{equation*}
  where $\{\,,\,\}_D^p$ is the bracket from Definition
  \ref{def:extended-dirac-bracket} associated with
  \begin{equation}\label{eq:trafo-quant}
    \left.
    \begin{aligned}
      &\bq_\psi^p=\Ad(g)\bq_\psi,\qquad \bq_\alpha^p=\Ad(g)\bq_\alpha,\qquad \bq_\theta^p=\Ad(g)(\bq_\theta-\bq_\alpha\wedge\bt),\\
      &r^p=\left(\Ad(p)\oo\Ad(p)\right)\left[r+\bar\eta^p-\bar\eta^p_{21}\right],
    \end{aligned}
    \qquad\right\}
  \end{equation}
  and $\bar\eta^p:\RR^2\to\poal\oo \poal$ is given by
  \begin{equation}\label{eq:trafo-eta}
    \bar\eta^p = q_\psi^a P_a \oo p^\inv\partial_\psi p + (q^a_\alpha  J_a+q^a_\theta P_a) \oo p^\inv\partial_\alpha p.
  \end{equation}
\end{lemma}

\begin{proof} $\quad$
\begin{enumerate}
\item To derive explicit expressions for the transformed bracket $\{,\}_D^p$,
  it is convenient to consider two cases separately, namely Lorentz
  transformations $g: \RR^2 \to \logr$, which do not depend on $\alpha$, and
  translations $t: \RR^2 \to \RR^3$, which depend on $\alpha$ at most linearly.

  We start by determining concrete expressions for $r^p$ from formula
  \eqref{eq:trafo-quant}.  For a Lorentz transformation $p=(g,0):\RR\to \logr$
  with $\partial_\alpha p=0$, we have $\bar\eta^p(\psi)=q^a_\psi(\psi) P_a \oo
  g^\inv\partial_\psi g(\psi)$. Expanding this terms of a basis as
  $g^\inv\partial_\psi g(\psi)=n^a(\psi)J_a$ then yields
  \begin{equation}\label{eq:r-gauged-classification-Ltrafo}
    r^p = \bigl(\Ad(g) \oo \Ad(g)\big) [r + q_\psi^a n^b (P_a \otimes J_b - J_b \otimes P_a)].
  \end{equation}
  In the case of a translation $p=(1,\bt):\RR^2\to\RR$ with
  $\partial_\alpha^2\bt=0$, we have $\bar\eta^p=q_\psi^a\partial_\psi t^b P_a
  \oo P_b + q_\alpha^a\partial_\alpha t^b J_a \oo P_b +
  q_\theta^a\partial_\alpha t^b P_a \oo P_b$. Inserting this with expression
  \eqref{eq:extended-dirac-r} for $r$ into \eqref{eq:trafo-quant} and using the
  identities $\Ad(\bt) J_a=J_a+\tensor{\ee}{_a_b^c} t^b P_c$,
  $\Ad(\bt)P_a=P_a$, we obtain after some computations
  \begin{equation}\label{eq:r-gauged-classification-translation}
    \begin{split}
    r^p = r &- \partial_\alpha t^a q_\alpha^b (P_a \otimes J_b - J_b \otimes P_a) \\
            &+ \tensor{\ee}{^a^b_c}\bigl[
                      (1-V^d_{\;\;d})\bt + V^T\bt
                          + \bq_\psi\wedge\partial_\psi\bt
                          + [\bq_\theta - \bq_\alpha\wedge\bt]\wedge\partial_\alpha\bt
               \bigr]^c P_a \otimes P_b.
    \end{split}
  \end{equation}

\item We now derive explicit expressions for the transformed Poisson brackets
  $\{\,,\,\}_D^p$.  For a Lorentz transformation $p=(g,0):\RR\to \logr$ with
  $\partial_\alpha p=0$, it follows directly from the identity
  $\{\psi,\alpha\}_D=0$ that the Poisson brackets involving the variables
  $\psi$ and $\alpha$ with functions on $\restpogr$ are given by:
  \begin{equation*}
    \begin{aligned}
      &\{\psi\circ \Phi^p, h\circ \Phi^p\}_D = 0, \\
      &\{\psi\circ \Phi^p, \bj_X\circ \Phi^p\}_D = \big[-\idadi{X} \Ad(g) \, \bq_\psi\big] \circ \Phi^p, \\
      &\{\alpha\circ \Phi^p, h\circ \Phi^p\}_D = \big[\smashoperator{\sum_{Y\in\{M_3,\dots,B_g\}}} \tensor{\Ad(g)}{^a_b}q_\alpha^b(J_{R,a}^Y+J_{L,a}^Y)h\big] \circ \Phi^p, \\
      &\{\alpha\circ \Phi^p, \bj_X\circ \Phi^p\}_D = \big[-\idadi{X} \Ad(g)\bq_\theta - (\Ad(g)\bq_\alpha) \wedge \bj_X\big] \circ \Phi^p,
    \end{aligned}
  \end{equation*}
  for all $h\in\cif(\restlogr)$ and $X\in\{M_3,\dots,B_g\}$.  To determine the
  brackets of the type $\{F, G\}_D^p$ with $F,G\in\cif(\restpogr)$, we again
  consider functions $h\in\cif(\restlogr)$ and the variables $j_X$,
  $X\in\{M_3,\dots,B_g\}$. After some straightforward computations, we obtain
  \begin{multline*}
    \{j_X^a\circ\Phi^p, h\circ \Phi^p\}_D = \\
      \{j_X^a, h\} \circ \Phi^p
      - \tensor{\Ad(g)}{^a_c} \tidadi{X}{^c_h} \bigl(\tensor{V}{^h_e} - q_\psi^h n_e\bigr) \Bigl[\smashoperator{\sum_{Y\in\{M_3,\dots,B_g\}}} (J^{R,e}_Y+J^{L,e}_Y)h\Bigr] \circ \Phi^p,
  \end{multline*}
  which allows us to express the transformed bracket as
  \begin{align*}
    &\{j_X^a\circ \Phi^p, h\circ \Phi^p\}_D
      = \Bigl[
          \{j_X^a, h\} -
          \tidadi{X}{^a_h} \tensor{({V}^p)}{^h_e} \smashoperator{\sum_{Y\in\{M_3,\dots,B_g\}}} (J^{R,e}_Y+J^{L,e}_Y)h
        \Bigr] \circ \Phi^p,\\[-.4em]
    &\text{with}\quad \tensor{({V}^p)}{^h_e} = \tensor{\Ad(g)}{^h_m}\tensor{\Ad(g)}{_e^p}\big[\tensor{V}{^m_p} - q_\psi^m n_p\big].
  \end{align*}
  An analogous calculation for the brackets of the form $\{j_X^a\circ \Phi^p,
  j_Y^b\circ \Phi^p\}_D$ shows that it is obtained by transforming $V\to V^p$
  as above and replacing $\bm:\RR^2\to\RR^3$ by
  $\bm^p=\Ad(g)\bm:\RR^2\to\RR^3$.  This implies that the transformed bracket
  takes the form $\{F\circ\Phi^p, G\circ\Phi^p\}_D = [\frrestbivector[r^p]
  (\diffd F \otimes \diffd G)]\circ\Phi^p$ with $r^p$ given by
  \eqref{eq:r-gauged-classification-Ltrafo} and proves the claim for the
  Lorentz transformations.

\item To determine the transformation of the bracket $\{\,,\,\}_D$ under
  translations, we again use the parametrisation in terms of the variables
  $j_X$, $X\in\{M_3,\dots,B_g\}$ and functions $h\in\cif(\restlogr)$. The
  transformation of these variables under a translation
  $p=(0,\bt):\RR^2\to\RR^3$ is given by $h\circ \Phi^p= h$ and $\bj_X \circ
  \Phi^p=\bj_X + \idadi{X}\bt$.  This implies directly that the brackets
  $\{\psi,h\}_D$ and $\{\alpha,h\}_D$ are preserved, and with the relations
  $\{\alpha,\psi\}_D=0$, $\{\psi, h\}_D=0$, one obtains the same for the
  brackets $\{\psi,\bj_X\}_D$.  The formula for the brackets
  $\{\alpha,\bj_X\}_D^p$ follows directly from the relation
  \begin{equation*}
    \Big\{\alpha, \tidadi{X}{^a_b}\Big\}_D =
      \Big[\tensor{\ee}{_d^a_m} \tidadi{X}{^m_b} +
      \tensor{\ee}{_d_b^m} \tidadi{X}{^a_m}\Big] q_\alpha^d.
  \end{equation*}
  We thus find that the transformed brackets involving the variables
  $\psi,\alpha$ are given by
  \begin{equation*}
    \begin{aligned}
      &\{\psi \circ \Phi^p, h\circ \Phi^p\}_D = 0, \\
      &\{\psi\circ \Phi^p, \bj_X\circ \Phi^p\}_D = \bigl[-\idadi{X} \, \bq_\psi\bigr] \circ \Phi^p, \\
      &\{\alpha\circ \Phi^p, h\circ \Phi^p\}_D = \bigl[\smashoperator{\sum_{Y\in\{M_3,\dots,B_g\}}} q_\alpha^a(J_{R,a}^Y+J_{L,a}^Y)h\bigr] \circ \Phi^p, \\
      &\{\alpha\circ \Phi^p, \bj_X\circ \Phi^p\}_D = \bigl[-\idadi{X} (\bq_\theta - \bq_\alpha \wedge \bt) - \bq_\alpha \wedge \bj_X\bigr] \circ \Phi^p.
    \end{aligned}
  \end{equation*}
  To determine the transformed brackets $\{F\circ \Phi^p, G\circ \Phi^p\}_D$
  for $F,G\in\cif(\restpogr)$, we calculate the brackets of functions
  $h\in\cif(\restlogr)$ and variables $\bj_X, \bj_Y$ for $X, Y\in
  \{M_3,\dots,B_g\}$. A direct computation yields
  \begin{align*}
    &\{j_X^a\circ \Phi^p, h\circ \Phi^p\}_D
      = \Bigl[\{j_X^a, h\} - \tidadi{X}{^a_g} \tensor{(V^p)}{^g_d} \smashoperator{\sum_{Y\in\{M_3,\dots,B_g\}}} (J^{R,d}_Y+J^{L,d}_Y)h\Bigr]\circ \Phi^p, \\
    \begin{split}
    &\{j_X^a\circ \Phi^p, j_Y^b\circ \Phi^p\}_D
      = \Bigl[\{j_X^a, j_Y^b\}
                      + \tidadi{X}{^a_g} \tensor{({V^p})}{^g_d}\,\tensor{\ee}{^d^b_f} j_Y^f\\[-.5em]
               &\quad - \tidadi{Y}{^b_g} \tensor{({V}^p)}{^g_d}\,\tensor{\ee}{^d^a_f} j_X^f
                      + \tidadi{X}{^a_c}\tidadi{Y}{^b_d} \ee^{cdf} m^p_f\Bigr] \circ \Phi^p,
    \end{split}\\
    &\text{with } \tensor{({V}^p)}{^b^c} = \tensor{V}{^b^c} + \partial_\alpha t^b q_\alpha^c,\\
    &\text{and } \bm^p= \bm+(1-\tensor{V}{^d_d}) \bt
                   + V^T \bt
                   + \bq_\psi \wedge \partial_\psi\bt
                   + (\bq_\theta - \bq_\alpha\wedge\bt)\wedge\partial_\alpha\bt.
  \end{align*}
  For all $F,G\in\cif(\restpogr)$ the transformed bracket therefore takes
  the form $\{F\circ \Phi^p, G\circ \Phi^p\}_D= [\frrestbivector[r^p] (\diffd F
  \otimes \diffd G)] \circ \Phi^p$ with $r^p$ given by
  \eqref{eq:r-gauged-classification-translation}.  This proves the claim.
\end{enumerate}
\end{proof}

Lemma \ref{lem:poinc-trafo} gives explicit expressions for the transformation
of the bracket in Definition~\ref{def:extended-dirac-bracket} under dynamical
Poincar\'e transformations which depend on the variables $\psi, \alpha$ and act
diagonally on $\restpogr$. It allows one to identify the transformed bracket
with another bracket of the form in Definition \ref{def:extended-dirac-bracket}
associated with transformed maps $r^p:\RR^2\to\poal\oo\poal$ and
$\bq_\psi^p,\bq_\alpha^p,\bq_\theta^p:\RR^2\to\RR^3$.

As the bracket $\{\,,\,\}_D$ in Definition \ref{def:extended-dirac-bracket} is
modelled after the Dirac bracket in Theorem \ref{thm:generic-dirac-bracket} and
Poincar\'e transformations of this type can be viewed as transitions between
different gauge fixing conditions, it is natural to ask whether these
Poincar\'e transformations preserve the Jacobi identity.  For this, note that,
given any Poisson manifold $(M, \{\,,\,\})$ and a diffeomorphism $\Phi: M \to
M$, one obtains a new Poisson bracket $\{\,,\,\}^\Phi$ on $M$ by setting $\{f,
g\}^\Phi \defeq \{f \circ \Phi, g \circ \Phi\} \circ \Phi^\inv$.  As shown in
Lemma \ref{lem:poinc-trafo}, the bracket $\{\,,\,\}^p_D = \{\,,\,\}^{\Phi^p}_D$
obtained from $\{\,,\,\}_D$ by applying the diffeomorphism $\Phi^p$ from
\eqref{eq:poincare-diffeo} is again of the form in Definition
\ref{def:extended-dirac-bracket}, but with the transformed maps $r^p,
\bq_\psi^p, \bq_\alpha^p, \bq_\theta^p$.  From Theorem
\ref{thm:jacobi-for-extended-dirac} we thus deduce that these transformed maps
are solutions of the CDYBE \eqref{eq:dcybe} and the additional conditions
\eqref{eq:q-relations} if and only if the original maps $r, \bq_\psi,
\bq_\alpha, \bq_\theta$ are.

\begin{corollary}
  Let $r: \RR^2\to \poal \oo \poal$ as in \eqref{eq:extended-dirac-r} be a
  solution of the CDYBE with $x^1=\psi$, $x^2=\alpha$, $x_1=q_\psi^a P_a$,
  $x_2=q_\alpha^a J_a+q_\theta^a P_a$ such that the conditions in
  \eqref{eq:q-relations} are satisfied and let $p:\RR^2\to \pogr$ be a
  dynamical Poincaré transformation as in Lemma \ref{lem:poinc-trafo}. Then
  $r^p:\RR^2\to \poal \oo \poal$ is a solution of the CDYBE with $x^1=\psi$,
  $x^2=\alpha$ and $x_1=q_\psi^{p,a}P_a$,
  $x_2=q_\alpha^{p,a}J_a+q_\theta^{p,a}P_a$ and satisfies
  \eqref{eq:q-relations}.  The map $\Phi^p$ is a Poisson isomorphism between
  the Poisson structures $\{\,,\,\}_D$ and $\{\,,\,\}_D^p$.
\end{corollary}

As discussed in the previous section, the equivalence of the CDYBE and
conditions \eqref{eq:q-relations} with the Jacobi identity for the bracket in
Definition \ref{def:extended-dirac-bracket} suggests that the solutions should
be viewed as a generalisation of the classical dynamical $r$-matrices in
Definition \ref{def:dyn-r-matrix} for which the associated abelian subalgebra
of $\poal$ is allowed to vary with the variables $\psi$ and $\alpha$. The
transformation formula \eqref{eq:trafo-quant} for the maps
$r:\RR^2\to\poal\oo\poal$ under dynamical Poincar\'e transformations and the
fact that these Poincar\'e transformations preserve the Jacobi identity
suggests that these Poincar\'e transformations should be interpreted as a
generalised version of the gauge transformations of classical dynamical
$r$-matrices introduced by Etingof and Varchenko in their work on the
classification of classical dynamical $r$-matrices \cite{Etingof:1998aa} (see
also \cite{Schiffmann:1998aa, Etingof:1999aa, Etingof:2000aa, Xu:2002aa} for
further work on the classification). We summarise the relevant definitions and
results from \cite{Etingof:1998aa}.

\begin{definition}[\cite{Etingof:1998aa}]\label{def:gauge-trafo}
Let $G$ be a Lie group, $H\subset G$ an abelian subgroup and
$r:\ah^*\to\ag\oo\ag$ a classical dynamical $r$-matrix for $(\ah,\ag)$. A
\definee{gauge transformation} of $r$ is a smooth function $\Pi: \ah^*\to G^H$
into the centraliser $G^H$ of $H$ in $G$ which acts on $r$ according
to\footnote{The sign difference between this formula and the one in
  \cite{Etingof:1998aa} is due to a different sign convention for the
  CDYBE~\eqref{eq:dcybe}.}
\begin{equation}\label{eq:g-trafo-r}
  r^{\Pi} = \bigl(\Ad(\Pi) \oo \Ad(\Pi)\bigr) [r + \bar{\eta}^\Pi - \bar{\eta}^\Pi_{21}],
\end{equation}
where $\bar \eta_\Pi:\ah^*\to \ah\oo \ag^{H}$ is the map dual to the
$\ag^\ah$-valued one-form $\eta_\Pi=\Pi^\inv\diffd\Pi$ on $\ah^*$ and
$\bar\eta^{21}_\Pi$ denotes its flip with values in $\ag^\ah\oo\ah$.
\end{definition}

The name gauge transformation is motivated by the fact that it maps classical
dynamical $r$-matrices for $(\ah,\ag)$ to classical dynamical $r$-matrices for
$(\ah,\ag)$. It is shown in \cite{Etingof:1998aa} that if $r$ is an
$\ah$-invariant solution of the CDYBE, then this also holds for the transformed
$r$-matrix $r^{\Pi}$.

\begin{theorem}[\cite{Etingof:1998aa}]
Let $G$ be a Lie group, $H\subset G$ an abelian subgroup and
$r:\ah^*\to\ag\oo\ag$ a classical dynamical $r$-matrix for $(\ah,\ag)$. Then
for every gauge transformation $\Pi:\ah^*\to G^H$, $r^\Pi$ is a classical
dynamical $r$-matrix for $(\ah,\ag)$.
\end{theorem}

By comparing formula \eqref{eq:trafo-quant} for the action of dynamical
Poincar\'e transformations on the solutions of the CDYBE with the one in
Definition \ref{def:gauge-trafo}, it becomes apparent that the two expressions
agree. The only difference is that in Definition \ref{def:gauge-trafo} the
gauge transformations are restricted to take values in the centraliser of the
subgroup $H\subset G$, while no such condition is imposed in our
case. Consequently, the dynamical Poincar\'e transformations also act on the
maps $\bq_\psi,\bq_\alpha,\bq_\theta:\RR^2\to \RR^3$ and hence on the
associated two-dimensional Lie subalgebra $\ah(\psi,\alpha)=\Span\{q_\psi^aP_a,
q_\alpha^aJ_a+q_\theta^aP_a\}$.  This is also apparent in the formula for the
map $\bar\eta^p: \RR^2\to \poal\oo\poal$ in \eqref{eq:trafo-eta}, which depends
on the chosen basis of the subalgebra $\ah(\psi,\alpha)$.

It is therefore instructive to consider the gauge transformations
\eqref{eq:g-trafo-r} for the classical dynamical $r$-matrices from Lemma
\ref{lem:dcybe-solutions} which are associated with fixed Cartan subalgebras
$\ah_a=\Span\{P_0,J_0\}\subset \poal$ and $\ah_b=\Span\{P_1,J_1\}\subset\poal$.
In that case, the abelian subgroup $H$ in Definition \ref{def:gauge-trafo} is
obtained by exponentiating, respectively, the Cartan subalgebras $\ah_a$ and
$\ah_b$, and the associated centraliser $G^H$ coincides with $H$.  With our
additional restriction that the Lorentzian component of $\Pi$ does not depend
on $\alpha$ and its translational component depends on $\alpha$ at most
linearly, the map $\Pi:\ah^*\to G^H$ in Definition \ref{def:gauge-trafo}
therefore takes the form $\Pi(\psi,\alpha)=(\exp(-\beta(\psi)J_j),
-[\gamma(\psi)+\alpha\delta(\psi)]\be_j)$ with $\beta,\gamma,\delta: \RR\to
\RR$ and $j=0$ in case $a$ and $j=1$ in case b). The transformation of the
classical dynamical $r$-matrices $r_{a,b}:\RR\to\poal\oo\poal$ in Lemma
\ref{lem:dcybe-solutions} under $\Pi$ is thus given by
\begin{equation*}
  r^\Pi(\psi,\alpha)=r(\psi,\alpha)-[\beta'(\psi)-\delta(\psi)] P_j\wedge J_j,
\end{equation*}
with $j=0$ in case a) and $j=1$ in case b). As $[\beta'(\psi)-\delta(\psi)]
P_j\wedge J_j$ satisfies the classical dynamical Yang-Baxter equation and
because $[[r(\psi,\alpha), P_j\wedge J_j]]+[[P_j\wedge J_j,
r(\psi,\alpha)]]=0$, it is directly apparent that this yields another classical
dynamical $r$-matrix for $(\ah,\poal)$ and only modifies $r$ by adding a twist.

Note in particular that $r_{a,b}$ are invariant under gauge transformations
$\Pi:\ah_{a,b}\to H$ of the form $\Pi(\psi,\alpha)=(\exp(-\beta(\psi)J_j),
-\alpha\beta'(\psi) P_j)$ and $\Pi(\psi,\alpha)=(\mathds{1},
-\gamma(\psi)P_j)$. The former correspond to combinations of a Lorentz
transformation that preserves $\ah_{a,b}$ and a translation in the direction of
its axis.  The latter correspond to translations which do not depend on the
parameter $\alpha$.  By specialising formula
\eqref{eq:extended-dirac-bracket-for-gauge-fixed} to the case at hand in which
$\bq_\psi=\bq_\alpha=e_j$, $\bq_\theta=0$ one finds that these are precisely
the flows that the variables $\psi\cdot \alpha$ and $\psi$ generate via the
bracket $\{\,,\,\}_D$ in Definition \ref{def:extended-dirac-bracket}.

\subsection{Standard solutions and classical dynamical r-matrices}

Although Theorem \ref{thm:jacobi-for-extended-dirac} provides a direct link
between the Jacobi identity for the bracket in Definition \ref{def:extended-dirac-bracket}
and solutions of the CDYBE that are subject to the additional conditions
\eqref{eq:q-relations}, the disadvantage of this description is that the
associated solutions $r:\RR^2\to\poal\oo\poal$ are in general quite complicated
and the additional conditions \eqref{eq:q-relations} do not have an immediate
geometrical interpretation.

It is therefore natural to ask if they can be related to a simple set of
standard solutions which define classical dynamical $r$-matrices in the sense
of Definition \ref{def:dyn-r-matrix}. The results of the previous subsection
suggest that this can be achieved by applying dynamical Poincar\'e
transformations.  As we will see in the following, this is possible for those
values of the parameters $\psi$ for which $\bq_\psi$ and $\bq_\alpha$ are
timelike or spacelike. A necessary and sufficient condition then is that the
maps $r:\RR^2\to\poal\oo\poal$ and $\bq_\psi,
\bq_\alpha,\bq_\theta:\RR^2\to\RR^3$ satisfy the equations
\eqref{eq:q-relations} in Theorem~\ref{thm:jacobi-for-extended-dirac}. We have
the following lemma:

\begin{lemma}\label{lem:standard-form}
Consider $r:I\times \RR\to\poal\oo\poal$ as in \eqref{eq:extended-dirac-r} and
$\bq_\psi,\bq_\alpha,\bq_\theta: I\times \RR\to\RR^3$, where $I\subset\RR$ is
an open interval with $\bq_\psi^2(\psi),\bq_\alpha^2(\psi)\neq 0$,
$\bq_\psi(\psi)\wedge\bq_\alpha(\psi)=0$ and
$\partial_\alpha\bq_\psi(\psi)=\partial_\alpha\bq_\alpha(\psi)=\partial_\alpha^2\bq_\theta(\psi)=0$
for all $\psi\in I$.  If $r$ and $\bq_\psi,\bq_\alpha,\bq_\theta$ satisfy the
conditions \eqref{eq:q-relations}, then there exists a Poincar\'e
transformation as in Lemma \ref{lem:poinc-trafo} such that the transformed
quantities $r^p: I \times \RR\to\poal\oo\poal$,
$\bq_\psi^p,\bq_\alpha^p,\bq_\theta^p:I\times \RR\to \RR^3$ defined by
\eqref{eq:trafo-quant} are of the form
\begin{equation}\label{eq:r-simple}
  r^p=\tfrac 1 2 (P_a\oo J^a+J^a\oo P_a)-\tfrac 1 2 \ee^{abc} w^p_c P_a\wedge J_b+\tfrac 1 2 \ee^{abc} m^p_c P_a\wedge P_b,
\end{equation}
with one of the following:
\begin{compactenum}[\quad a)]
\item $\bq_\psi^p,\bq^p_\alpha,\bq^p_\theta,\bw^p,\bm^p\in\Span\{\be_0\}$ and
  $q_\alpha^{p,0}\partial_\alpha q_\theta^{p,0}=q_\psi^{p,0}\partial_\psi
  q_\alpha^{p,0}$,

\item $\bq_\psi^p,\bq_\alpha^p,\bq_\theta^p,\bw^p,\bm^p\in\Span\{\be_1\}$ and
  $q_\alpha^{p,1}\partial_\alpha q_\theta^{p,1}=q_\psi^{p,1}\partial_\psi
  q_\alpha^{p,1}$.
\end{compactenum}
\end{lemma}

\begin{proof} $\quad$
\begin{enumerate}
\item Let $r$ be of the form \eqref{eq:extended-dirac-r} and
  $\bq_\psi,\bq_\alpha,\bq_\theta:I\times\RR\to\RR^3$ with
  $\bq_\psi\wedge\bq_\alpha=0$ and either $\bq_\psi^2,\bq_\alpha^2>0$ (case a)
  or $\bq_\psi^2,\bq_\alpha^2<0$ (case b).  Then the formulas for the action of
  dynamical Lorentz transformations in Lemma \ref{lem:poinc-trafo} imply that
  via a suitable Lorentz transformation $g:\RR^2\to \logr$, $\partial_\alpha
  g=0$, we can achieve one of the following: a)
  $\bq_\psi,\bq_\alpha\in\Span\{\be_0\}$ or b)
  $\bq_\psi,\bq_\alpha\in\Span\{\be_1\}$.

  The resulting matrix $V: \RR^2 \to \Mat(3, \RR)$ in
  \eqref{eq:extended-dirac-r} can be decomposed into a symmetric and an
  antisymmetric component according to
  \begin{equation*}
    V^{ab}(\psi)=\tfrac 1 2 \eta^{ab}+ Q^{ab}(\psi)+\tfrac 1 2 \ee^{abc} w_c(\psi)
  \end{equation*}
  with $\bw:I\to\RR^3$ and $Q^{ab}:I\to\Mat(3, \RR)$ symmetric.  By applying a
  suitable translation $\bt: I\times \RR\to\RR^3$, $\partial_\alpha^2 \bt=0$,
  which does not affect $\bq_\psi,\bq_\alpha$, one can then achieve that the
  symmetric matrix $Q^{ab}:I\times \RR\to\Mat(3, \RR)$ satisfies
  $Q^{0a}=Q^{a0}=0$ $\forall a\in\{0,1,2\}$ in case a) and $Q^{1a}=Q^{a1}=0$
  $\forall a\in\{0,1,2\}$ in case b). With a further translation $\bt':\RR^2\to
  \RR^3$ which satisfies $\partial_\alpha\bt'=0$ and hence does not affect $V$,
  one can achieve that $\bq_\theta$ takes the form
  $\bq_\theta=\alpha\cdot \partial_\alpha\bq_\theta+\tilde\bq_\theta$ with
  $\partial_\alpha\tilde\bq_\theta=0$ and $\tilde\bq_\theta\in\Span\{\be_0\}$
  in case a) or $\tilde\bq_\theta\in\Span\{\be_1\}$ in case b).

\item After these transformations, the first condition in
  \eqref{eq:q-relations} is satisfied if and only if $\bw\in\Span\{\be_0\}$ and
  $Q^{11}+Q^{22}=0$ in case a) and $\bw\in\Span\{\be_1\}$ and $Q^{00}-Q^{22}=0$
  in case b).  Under these conditions, the second equation in
  \eqref{eq:q-relations} simplifies to
  \begin{equation*}
    q^a_\alpha\partial_\alpha q_\theta^b-q_\psi^b\partial_\psi q^a_\alpha
    +q_\alpha^d(\ee^{a}_{\;\;dh} Q^{bh}+\ee^{b}_{\;\;dh} Q^{ah})=0\qquad\forall a,b\in\{0,1,2\}.
  \end{equation*}
  This implies $\partial_\alpha\bq_\theta\in\Span\{\be_0\}$, $Q=0$ in case a)
  and $\partial_\alpha\bq_\theta\in\Span\{\be_1\}$, $Q=0$ in case b).
  Combining this with the previous results, we obtain
  $\bq_\theta\in\Span\{\be_0\}$ and $q_\alpha^0 \partial_\alpha
  q_\theta^0=q_\psi^0\partial_\psi q_\alpha^0$ in case a) and
  $\bq_\theta\in\Span\{\be_1\}$ and $q_\alpha^1 \partial_\alpha
  q_\theta^1=q_\psi^1\partial_\psi q_\alpha^1$ in case b).  Inserting these
  conditions into the third equation in \eqref{eq:q-relations}, one finds that
  this equation simplifies to the condition $\bm \wedge\bq_\alpha=0$, which
  implies $\bm\in\Span\{\be_0\}$ in case a) and $\bm\in\Span\{\be_1\}$ in case
  b).  This proves the claim.
\end{enumerate}
\end{proof}

If $r:\RR^2\to\poal\oo\poal$ and $\bq_\psi,\bq_\alpha,\bq_\theta:
\RR^2\to\RR^3$ are of the form in Lemma \ref{lem:standard-form}, then
$\ah(\psi,\alpha)=\Span\{q_\psi^a P_a, q^a_\alpha J_a+q_\theta^aP_a\}$ is a
fixed Cartan subalgebra of $\poal$ which no longer varies with $\psi$ and
$\alpha$. Moreover, a direct calculation shows that $r(\psi,\alpha)$ is then
invariant under the action of the Cartan subalgebra $\ah(\psi,\alpha)$: $
[\by\oo 1+1\oo \by, r(\psi,\alpha)]=0$ for all $\by\in \ah(\psi,\alpha)$.

This provides us with a natural geometrical interpretation of the conditions
\eqref{eq:q-relations} in Theorem~\ref{thm:jacobi-for-extended-dirac}. These
conditions ensure that for all values of $\psi$ for which
$\bq_\psi^2(\psi),\bq_\alpha^2(\psi)\neq 0$, the maps
$r:\RR^2\to\poal\oo\poal$, $\bq_\psi,\bq_\alpha,\bq_\theta:\RR^2\to\RR^3$
(which are {not} required to satisfy the CDYBE in Lemma
\ref{lem:standard-form}) can be brought into a standard form via a suitable
dynamical Poincar\'e transformation. The resulting map
$r:\RR^2\to\poal\oo\poal$ is then invariant under the fixed Cartan subalgebra
spanned by $\bq_\psi,\bq_\theta,\bq_\alpha$.  The conditions
\eqref{eq:q-relations} can therefore be viewed as a generalised or
Poincar\'e-transformed version of the restriction to a fixed Cartan subalgebra
in Definition \ref{def:dyn-r-matrix}.  As shown in the proof of Theorem
\ref{thm:jacobi-for-extended-dirac}, they ensure that the Jacobi identity holds
for mixed brackets involving functions of the variables $\psi,\alpha$ and
functions on $\restpogr$.

Lemma \ref{lem:standard-form} applies in particular to solutions of the CDYBE
that satisfy the conditions \eqref{eq:q-relations} and hence give rise to
Poisson structures $\{\,,\,\}_D$ on $\RR^2\times\restpogr$. This allows one to
(locally) classify all possible solutions of the CDYBE that arise from gauge
fixing conditions satisfying the conditions in Section
\ref{subsec:constraints-and-gauge-fixing} and hence all associated Dirac
brackets.

\begin{theorem}\label{thm:dcybe-standard-transformation}
Let $r:I\times \RR\to\poal\oo\poal$ be a solution of the CDYBE with $x^1=\psi$,
$x^2=\alpha$, $x_1=q_\psi^aP_a$, $x_2=q_\alpha^aJ_a+q_\theta^aP_a$ that
satisfies conditions \eqref{eq:q-relations} in Theorem
\ref{thm:jacobi-for-extended-dirac} and for which $\bq_\psi^2,\bq_\alpha^2\neq
0$, $\bq_\psi\wedge\bq_\alpha=0$ and
$\partial_\alpha\bq_\psi=\partial_\alpha\bq_\alpha=\partial_\alpha^2\bq_\theta=0$
on $I\times\RR$. Then there exists a Poincar\'e transformation $p:I\times\RR\to
\pogr$ as in Lemma \ref{lem:standard-form} and a diffeomorphism $\by=(y^1,
y^2): I\times\RR\to I'\times\RR$ with $\partial_\alpha y_1 = 0$ and
$\partial_\alpha^2 y_2=0$ such that one of the following holds:
\begin{compactenum}[\quad a)]
\item $\bq_\psi,\bq_\alpha,\bq_\theta \in \Span\{\be_0\}$ and
  \vspace{-0.3cm}
  \begin{equation*}
    r^p(\psi,\alpha)=\tfrac 1 2 (P_a\!\oo\!J^a \!+\! J^a\!\oo\!P_a) + \tfrac12 \tan\tfrac{y^1(\psi)}{2} \left(P_1\!\wedge\!J_2 \!-\! P_2\!\wedge\!J_1\right) + \frac{y^2(\psi,\alpha)}{4\cos^2\tfrac{y^1(\psi)}{2}}P_1\!\wedge\! P_2,
  \end{equation*}
  \vspace{-0.4cm}

\item $\bq_\psi,\bq_\alpha,\bq_\theta \in \Span\{\be_1\}$ and
  \vspace{-0.3cm}
  \begin{equation*}
    r^p(\psi,\alpha)=\tfrac 1 2 (P_a\!\oo\!J^a \!+\! J^a\!\oo\!P_a) + \tfrac12 \tanh\tfrac{y^1(\psi)}{2} \left(P_2\!\wedge\!J_0 \!-\! P_0\!\wedge\!J_2\right) + \frac{y^2(\psi,\alpha)}{4\cosh^2\tfrac{y^1(\psi)}{2}}P_2\!\wedge\! P_0.
  \end{equation*}
\end{compactenum}
\end{theorem}

\begin{proof} $\quad$
\begin{enumerate}

\item By Lemma \ref{lem:standard-form} there exists a Poincar\'e-valued
  function $p:I\times \RR\to \pogr$ whose Lorentzian component does not depend
  on $\alpha$ and whose translation component depends on $\alpha$ at most
  linearly such that $r^p:I\times \RR\to\poal\oo\poal$ and
  $\bq_\psi^p,\bq_\alpha^p,\bq_\theta^p:I\times \RR\to\RR^3$ satisfy again the
  CDYBE as well as conditions \eqref{eq:q-relations} and are of the form given
  in Lemma~\ref{lem:standard-form}.  This implies that $r$ is of the form
  \eqref{eq:r-simple}:
  \begin{align*}
    r^p=\tfrac 1 2 (P_a\oo J^a+J^a\oo P_a)-\tfrac 1 2 \ee^{abc} w^p_c P_a\wedge J_b+\tfrac 1 2 \ee^{abc} m^p_c P_a\wedge P_b
  \end{align*}
  with $\bq_\psi^p,\bq_\alpha^p,\bq_\theta^p,\bw^p,\bm^p$ subject to conditions
  a) or b) in Lemma \ref{lem:standard-form}.  It follows that there are
  smooth functions $\beta,\gamma,\delta,\epsilon, \varphi_0,\varphi_1: I \to
  \RR$ with $\beta(\psi),\gamma(\psi)\neq 0$ for all $\psi\in I$ such that
  \begin{equation*}
    \bq^p_\psi=\beta \be_j,\quad
    \bq^p_\alpha=\gamma \be_j,\quad
    \bq^p_\theta=\left(\delta+{\alpha\, \gamma'\beta}/{\gamma}\right)\be_j,\quad
    \bw^p=\epsilon\,\be_j,\quad
    \bm^p=(\varphi_0+\alpha\varphi_1)\be_j,
  \end{equation*}
  where $j=0$ in case a) and $j=1$ in case b).  Inserting these expressions
  into \eqref{eq:dcybe-with-simple-UV}, one finds that $r^p$ satisfies the
  CDYBE with $x^1=\psi$, $x^2=\alpha$, $x_1=q_\psi^{p,a}P_a$ and
  $x_2=q_\alpha^{p,a}J_a+q_\theta^{p,a}P_a$ if and only if the coefficient
  functions $\beta,\gamma,\delta,\epsilon, \varphi_0,\varphi_1: I \to \RR$
  satisfy the following set of differential equations:
  \begin{align}
    &\gamma \varphi_1+\tfrac 1 2 \beta\epsilon'=0,\qquad
    1\pm\epsilon^2\pm 2\beta\epsilon'=0,\qquad
    \beta\varphi_0'+\delta\varphi_1+\epsilon\varphi_0=0,
  \end{align}
  where the sign $+$ in the second equation refers to case a), the sign $-$ to
  case b).  Set $g=1/\gamma$ and let $f:\RR\to\RR$ be a function with
  $f'=1/\beta$. Then the second equation can be integrated to $\epsilon(\psi)=-
  \tan(f(\psi)/2)$ in case a) and $\epsilon(\psi)=\tanh(f(\psi)/2)$ in case
  b). Inserting this result into the remaining two equations, we obtain
  \begin{equation*}
    \varphi_1(\psi)=-\tfrac 1 2 g(\psi)\epsilon'(\psi),\qquad
    \delta(\psi)=-\frac{h'(\psi)}{f'(\psi) g(\psi)},
  \end{equation*}
  with $h(\psi)=4\varphi_0(\psi)\cos^2(f(\psi)/2)$ in case a) and
  $h(\psi)=4\varphi_0(\psi)\cosh^2(f(\psi)/2)$ in case b).  It follows that
  $\bq_\psi,\bq_\alpha,\bq_\theta,\bw,\bm$ are given by:
  \begin{equation}\label{eq:transf-qs}
  \left.
  \begin{aligned}
    &\bq_\psi^p(\psi)=\frac {\be_0}{f'(\psi)},\quad
     \bq_\alpha^p(\psi)=\frac {\be_0}{g(\psi)},\quad
     \bq_\theta^p(\psi,\alpha)=-\frac{h'(\psi)+\alpha g'(\psi)}{f'(\psi) g(\psi)}\,\be_0,\\
    &\bw^p(\psi)=-\tan(f(\psi)/2)\,\be_0,\quad\;\bm^p(\psi,\alpha)=\frac{h(\psi)+\alpha g(\psi)}{4\cos^2(f(\psi)/2)}\,\be_0\quad\;\;\text{in case a)},\\
    &\bw^p(\psi)=\tanh(f(\psi)/2)\,\be_0,\quad\bm^p(\psi,\alpha)=-\frac{h(\psi)+\alpha g(\psi)}{4\cosh^2(f(\psi)/2)}\,\be_0\quad\text{in case b)}.
  \end{aligned}
  \quad\right\}
  \end{equation}

\item With the definitions $y^1(\psi)=f(\psi)$,
  $y^2(\psi,\alpha)=h(\psi)+\alpha g(\psi)$ this yields expressions a) and b)
  for $r$.  For any $y_1,y_2\in \poal$ the right-hand side of the CDYBE then
  takes the form
  \begin{align*}
    &y_1^{(1)}\partial_{y^1}\,r_{23}- y_1^{(2)}\partial_{y^1}\, r_{13}+y_1^{(3)}\partial_{y^1}\, r_{12}+y_2^{(1)}\partial_{y^2}\,r_{23}- y_2^{(2)}\partial_{y^2}\, r_{13}+y_2^{(3)}\partial_{y^2}\, r_{12}\\
    &=x_1^{(1)}\partial_{\psi}\,r_{23}- x_1^{(2)}\partial_{\psi}\, r_{13}+x_1^{(3)}\partial_{\psi}\, r_{12}+x_2^{(1)}\partial_{\alpha}\,r_{23}- x_2^{(2)}\partial_{\alpha}\, r_{13}+x_2^{(3)}\partial_{\alpha}\, r_{12},
  \end{align*}
  where
  \begin{equation*}
    x_1(\psi)=\frac{y_1(\psi)}{f'(\psi)},\qquad
    x_2(\psi, \alpha)=\frac{y_2(\psi, \alpha)}{g(\psi)}-\frac{(h'(\psi)+\alpha g'(\psi))\,y_1(\psi)}{f'(\psi)g(\psi)}.
  \end{equation*}
  Setting $y_1=P_0$, $y_2=J_0$ in case a) and $y_1=P_1$, $y_2=J_1$ in case b),
  we obtain $x_1=q_\psi^{p,a}P_a$, $x_2=q_\alpha^{p,a}J_a+q_\theta^{p,a}P_a$
  with $\bq_\psi^p,\bq_\alpha^p,\bq_\theta^p$ given by
  \eqref{eq:transf-qs}. This proves the claim.
\end{enumerate}

\end{proof}

Theorem \ref{thm:dcybe-standard-transformation} amounts to a classification of
all possible solutions of the CDYBE \eqref{eq:dcybe} of the form in Definition
\ref{def:extended-dirac-bracket} that satisfy the additional conditions
\eqref{eq:q-relations} and hence of all Poisson structures of the type in
Definition \ref{def:extended-dirac-bracket}. It states that locally all such
Poisson structures are obtained from one of the classical dynamical
$r$-matrices in Lemma \ref{lem:dcybe-solutions} by applying a suitable
Poincar\'e transformation together with a rescaling of the variables $\psi$,
$\alpha$.  Note, however, that this classification is only local in the
following sense: for any given value $\psi_0$ of the variable $\psi$ for which
$\bq_\psi^2,\bq_\alpha^2\neq 0$, there is an open interval $I\subset\RR$,
$\psi_0\in I$, such that $r(\psi,\alpha)$ can be transformed into one of the
classical dynamical $r$-matrices in Theorem
\ref{thm:dcybe-standard-transformation} for all $(\psi,\alpha)\in I\times\RR$.

In particular, this locally determines all possible Dirac brackets that arise
from gauge fixing procedures that satisfy the well-motivated conditions in
Section \ref{subsec:constraints-and-gauge-fixing}. The Dirac bracket of such a
gauge fixing procedure is always determined by a solution of the CDYBE that
satisfies the additional conditions \eqref{eq:q-relations} and
$\bq_\psi\wedge\bq_\alpha=0$. For those values of the variable $\psi$ for which
$\bq_\psi^2(\psi),\bq_\alpha^2(\psi)\neq 0$ , the resulting Dirac bracket can
be transformed into the bracket defined by one of the two classical dynamical
$r$-matrices in Lemma \ref{lem:dcybe-solutions}.

However, the Dirac bracket resulting from a generic gauge fixing condition is
associated with maps $\bq_\psi,\bq_\alpha:\RR\to\RR^3$ for which the signature
of $\bq_\psi^2,\bq_\alpha^2$ changes as a function of the variable $\psi$.  In
particular, this is the case for the specific gauge fixing conditions
investigated in \cite{Meusburger:2011aa}. It is shown there that the map
$\bq_\psi:\RR^2\to\RR^3$ is timelike, lightlike or spacelike for those values
of $\psi$ for which, respectively, the product $u_{M_2}\cdot u_{M_1}$ of the
Lorentzian components of the gauge-fixed holonomies is elliptic, parabolic or
hyperbolic.  It is well-known that the product of two elliptic elements of
$\logr$ can be either elliptic, parabolic or hyperbolic. The first requirement
on the gauge fixing conditions in Section \ref{subsec:dirac-constraints}
implies that all of these cases must arise in the Dirac bracket. This suggests
that such transitions between timelike, spacelike and lightlike solutions are a
generic outcome of the gauge fixing procedure when the two holonomies $M_1,M_2$
are elliptic.

It is therefore not possible to reduce the investigation of gauge fixing
procedures and the resulting Poisson structures to classical dynamical
$r$-matrices in the sense of Definition \ref{def:dyn-r-matrix}.  Instead, one
needs to admit more general solutions of the CDYBE and to allow the associated
Lie subalgebras $\ah(\psi,\alpha)\subset \poal$ to vary non-trivially with the
variables $\psi$ and $\alpha$. Such solutions are no longer invariant under the
action of the subalgebra $\ah(\psi,\alpha)$. Instead, they satisfy the
generalised consistency conditions \eqref{eq:q-relations} which together with
the CDYBE ensure the Jacobi identity for the associated Poisson bracket.

\section{Outlook and conclusions}
\label{sec:outlook}

In this paper we applied the Dirac gauge fixing procedure to the description of
the moduli space of flat $\pogr$-connections in terms of an ambient space with
an auxiliary Poisson structure \cite{Fock:1998aa}.  We investigated a large
class of gauge fixing conditions subject to two well-motivated structural
requirements, namely that the gauge fixing is based on the choice of two
punctures on the underlying Riemann surface and that it preserves the natural
$\NN$-grading of the auxiliary Poisson structure.

We showed that the Poisson algebras obtained from gauge fixing are in
one-to-one correspondence with solutions of the classical dynamical Yang-Baxter
equation (CDYBE) and of a set of additional equations. The latter can be viewed
as the counterpart of the usual requirement of invariance of the classical
dynamical $r$-matrices under the action of a fixed Cartan subalgebra.  These
solutions of the CDYBE are a generalisation of classical dynamical $r$-matrices
in which the associated two-dimensional Lie subalgebras of $\poal$ vary with
the dynamical variables.

We also demonstrated how a change of gauge fixing conditions affects the
associated solutions of the CDYBE and showed that this change corresponds to
the action of dynamical Poincar\'e transformations.  These dynamical Poincar\'e
transformations generalise the gauge transformations of classical dynamical
$r$-matrices in \cite{Etingof:1998aa}.  We found that every solution obtained
via gauge fixing can be transformed into one of two classical dynamical
$r$-matrices via these transformations and a rescaling for almost all values of
the dynamical variables.  This gives rise to a complete (local) classification
of all possible outcomes of gauge fixing in the context of the moduli space of
flat $\pogr$-connections.

However, for generic solutions there are also certain values of the dynamical
variables for which the solutions cannot be transformed into standard classical
dynamical $r$-matrices.  These singular points appear at the transition between
classical dynamical $r$-matrices for two non-conjugate Cartan subalgebras of
$\poal$ and are associated with two-dimensional Lie subalgebras of $\poal$
which contain parabolic elements of $\loal$. This occurs for instance in the
solutions in Lemma \ref{lem:dcybe-for-special-r} and cannot be excluded by
suitable gauge fixing conditions.

To our knowledge this phenomenon does not appear in earlier references on the
topic. We expect that this is due to the fact that most of these references
investigate classical dynamical $r$-matrices for complex (semi)simple Lie
groups, for which all Cartan subalgebras are conjugate.  Based on our results,
we would predict that such transitions between classical dynamical $r$-matrices
for non-conjugate Cartan subalgebras arise for real Lie groups for which the
underlying symmetric, non-degenerate $\Ad$-invariant symmetric form
$\langle\,,\,\rangle$ is indefinite.

Another aspect which merits further investigation is the physical
interpretation of these classical dynamical $r$-matrices in the context of
(2+1)-dimensional (quantum) gravity.  It was shown in \cite{Meusburger:2011aa}
that gauge fixing procedures of the type investigated in this paper can be
viewed as the specification of an observer in a (2+1)-dimensional
spacetime. Moreover, the results in \cite{Meusburger:2011aa} suggest a direct
geometrical interpretation for the two dynamical variables in the classical
dynamical $r$-matrices: they correspond to the total mass and angular momentum
of the spacetime as measured by this observer. In this interpretation, the two
standard classical dynamical $r$-matrices would be associated with the
centre-of-mass frame of the universe and the transition points between
different Cartan subalgebras would correspond to the formation of Gott pairs
\cite{Gott:1991aa}.

We expect that our results could be generalised to the other moduli spaces of
flat connections that arise in the description of (2+1)-gravity for different
signatures and different values of the cosmological constant.  For Lorentzian
signature, the relevant Lie groups are $\pogr$, $\PSL(2,\RR)\times \PSL(2,\RR)$
and $\PSL(2,\CC)$, respectively, for vanishing, negative and positive
cosmological constant.  For Euclidean signature, the corresponding Lie groups
are $\ISO(3)$, $\SO(3)\times\SO(3)$ and $\PSL(2,\CC)$. It seems plausible that
analogous gauge fixing procedures applied to these groups would yield similar
outcomes, with the exception of the transition between different Cartan
subalgebras, which should not occur for Euclidean signature.

In this context, it would also be desirable to understand in more detail how
our results are related to the classical dynamical $r$-matrix symmetries
obtained by Buffenoir, Noui and Roche via a regularisation procedure for point
particles coupled to Chern-Simons theory with gauge group $\SL(2,\CC)$
\cite{Buffenoir:2003aa, Buffenoir:2005aa, Noui:2005aa}.  Although the approach
and setting in this work are very different, there should be an underlying
reason that forces the appearance of classical dynamical $r$-matrices in both
cases.

It would also be interesting to investigate in more detail the relation between
gauge-fixed Poisson structures on the moduli space of flat $\pogr$-connections
and mathematical structures associated with classical dynamical $r$-matrices
such as Poisson-Lie groupoids.  As the auxiliary Poisson algebra which is
gauge-fixed is closely related to certain structures from the theory of
Poisson-Lie groups, namely the dual Poisson-Lie structure and the Heisenberg
double, it could be expected that gauge fixing should be related to the
construction of dynamical versions of these structures.

Finally, we expect our results to have useful application in the
quantisation of the moduli space of flat $\pogr$-connections. This is due to
the fact that the resulting Poisson structure is very closely related to Fock
and Rosly's Poisson structure on the ambient space. The only difference is that
the classical $r$-matrix is replaced by a classical dynamical $r$-matrix
associated with two-dimensional Lie subalgebras of $\poal$. As Fock and Rosly's
Poisson structure serves as the starting point for the combinatorial
quantisation formalism, this suggests that this formalism could be extended
straightforwardly to include the gauge-fixed Poisson structure. This would
reduce the task of quantising the theory to the construction of the associated
dynamical quantum group.

\section*{Acknowledgements}

This work was supported by the Emmy Noether research grant ME 3425/1-1 of the
German Research Foundation (DFG).  We thank  Karim Noui for
useful remarks and discussions and Stefan Waldmann and  Winston Fairbairn for comments on a draft of
this paper.  The \textsc{xAct} tensor calculus package
\cite{Martin-Garcia:2008aa} proved useful in some of the computations.

\bibliography{dcybe}{}
\bibliographystyle{custom}

\end{document}